\documentclass{article}

\usepackage{times}
\usepackage[pdftex]{graphicx}
\usepackage{amssymb,amsfonts,amsmath,amsthm}
\usepackage{float}
\usepackage{setspace}

\newtheorem{theorem}{Theorem}

\newtheorem{definition}[theorem]{Definition}

\newtheorem{lemma}[theorem]{Lemma}

\newtheorem{proposition}[theorem]{Proposition}
\newtheorem{remark}[theorem]{Remark}

\numberwithin{theorem}{section}

\newcommand\EE {\mathbb E}

\newcommand\RR {\mathbb R}

\def\bone{\mathbf{1}}

\def\tmu{\tilde{\mu}}

\def\tV{\tilde{V}}

\begin{document}

\title{ Optimal investment for all time horizons and Martin boundary of space-time diffusions}
\author{Sergey Nadtochiy\thanks{
sergeyn@umich.edu; Department of Mathematics, University of Michigan.} \  and Michael Tehranchi\thanks{
m.tehranchi@statslab.cam.ac.uk; Statistical Laboratory, University of Cambridge.}
\thanks{
The results of this paper were first presented at the conference Advances in Portfolio Theory and Investment Management, in the Oxford-Man Institute, Oxford, UK, in May 2011. Both authors acknowledge the role of this meeting in the exchange of ideas that resulted in this work.}}
\date{This version: December 17, 2013\\ First draft: October 5, 2012}
\maketitle

\begin{abstract}
This paper is concerned with the axiomatic foundation and explicit construction of a general class of optimality criteria that can be used for investment problems with multiple time horizons, or when the time horizon is not known in advance. Both the investment criterion and the optimal strategy are characterized by the Hamilton-Jacobi-Bellman equation on a semi-infinite time interval. In the case when this equation can be linearized, the problem reduces to a time-reversed parabolic equation, which cannot be analyzed via the standard methods of partial differential equations. Under the additional uniform ellipticity condition, we make use of the available description of all minimal solutions to such equations, along with some basic facts from potential theory and convex analysis, to obtain an explicit integral representation of all positive solutions. These results allow us to construct a large family of the aforementioned optimality criteria, including some closed form examples in relevant financial models. 
\end{abstract}

{\bf Keywords:} \begin{small} Preferences, state-dependent utility, time-consistency, forward performance process, time-reversed HJB equation, Widder's theorem, Martin boundary.\end{small}


\section{Introduction}

The classical investment problem (also known as the Merton's problem) is concerned with the optimal allocation of investor's capital among available financial instruments.
The precise understanding of this statement depends on the notion of optimality employed by the decision maker. We consider the optimality criteria that are based on the characteristics of the terminal wealth generated by each strategy. In the academic literature, these characteristics are, usually, summarized in the expectation of a \emph{utility} function of the terminal wealth. More precisely, the investor (agent) chooses a utility function, along with an investment horizon, say $T$, and maximizes the expectation of this function applied to the terminal wealth payoff at time $T$ (represented by a random variable on some probability space), over all attainable payoffs\footnote{See \cite{Ekeland}, for an equilibrium approach, which does not require an optimality criterion.}. One of the main advantages of this approach is the existence of its \emph{axiomatic justification}. Assume that investor has \emph{preferences} on the set of all terminal payoffs (random variables, or, distributions), which form a \emph{complete order}: for any given pair of payoffs, the investor either prefers one to the other, or is indifferent between the two (cf. \cite{Cantor}). Then, the celebrated Von NeumannÐMorgenstern theorem (cf. \cite{vNM}) shows that, if this complete order satisfies several intuitive axioms, it has to be represented by an expected utility. In other words, there exists a utility function, such that, between any two payoffs, the investor always prefers the one with larger expected utility. There exist several variations in the choice of the axioms and in the properties of the resulting utility functions: see, for example, \cite{Bernoulli1738}, \cite{deFinetti}, \cite{Savage}, \cite{Jensen}. However, the most common set of axioms is the one due to Von Neumann and Morgenstern, and it consists of \emph{transitivity}, \emph{continuity} and \emph{independence} (cf. \cite{Jensen}). The \emph{risk aversion} axiom is often added to ensure that the diversification of a portfolio is encouraged in the resulting optimal investment problem and, in particular, the associated utility function is \emph{concave}. Once the set of axioms is chosen, we may assume, without loss of generality, that the investor's preferences on the set of terminal payoffs are determined by a utility function. Having chosen the appropriate utility function, we, then, solve the associated stochastic optimization problem to find the optimal strategy. Such problems have been widely studied under rather general assumptions on the market model and constitute one of the most active areas of research in modern theory of mathematical finance (see, for example, \cite{merton}, \cite{merton-a}, \cite{KramkovSchacher}, \cite{KramkovSirbu}, \cite{KarLehShreve}, \cite{Pliska}).

In a model where the investment decision is only made once, the outcome of agent's decision is a global trading strategy, which runs up until the terminal time horizon. Then, the optimal strategy is chosen at the initial time as the one that maximizes the expected utility of terminal wealth. However, such a definition of optimal strategy is not natural if the investment decisions are made at multiple times. Indeed, in the latter case, the outcome of every decision is a local investment strategy, which prescribes the actions in the next time period only and results in a random set of future investment opportunities, rather than in terminal payoff itself. Therefore, at each decision time, the agent needs to have a family of preferences on the associated space of set-valued random variables. The resulting family of dynamic preferences, also, has to be non-contradictory across time, or \emph{time-consistent}. Put simply, time-consistency means that the investor ``does not regret" her past decisions. It is best described in \cite{KrepsPorteus}, where time-consistency is postulated as one of the axioms, and the representation of all dynamic preferences satisfying these axioms, also known as \emph{recursive utilities}, is developed. In the context of expected utility, it is natural to construct the dynamic preferences by evaluating each local strategy (investment plan over the next time period) as the maximum expected utility of terminal wealth over all future (global) strategies that coincide with the chosen local strategy in the next time period. The \emph{dynamic programming principle}, when it holds, ensures the time-consistency of the resulting family of dynamic preferences. In fact, it also shows that the single-decision utility maximization problem (where the global strategy is chosen at the very beginning) is equivalent to a time-consistent multi-decision optimization problem, in which the optimality of the strategy is re-evaluated at each decision time. Such an equivalence (i.e. the dynamic programming principle) turns the problem of optimal investment into a stochastic control problem, described, for example, in \cite{FlemingSoner}, \cite{KrylovOptim}.

Despite the presence of an axiomatic foundation and the existence of the dynamic programming principle, the optimality criterion based on maximum expected utility has significant limitations. One of its biggest shortcomings is the fact that only the wealth payoff at a fixed time $T$ is taken into account when making the investment decision. In practice, one may want to consider additional properties of the wealth process: for example, its marginal distributions at {\bf all} time horizons $T>0$. The latter choice may be reasonable if, for example, the terminal time horizon is not known in advance. It is well known that expected utility cannot be easily generalized to the case of unbounded time horizons (except for some specific constructions). To illustrate the difficulty, assume that investor has chosen a time horizon $T$, along with a utility function $U$, and has solved the resulting optimization problem obtaining the optimal investment strategy on the time interval $[0,T]$. Assume, further, that ``life does not end" at $T$. Then, the investor chooses a longer time horizon $T'>T$, along with a new utility function $U'$, and constructs the optimal strategy on $[T,T']$. However, by doing this, the investor would like to ensure that her present decisions do not contradict the future ones. In other words, $U'$ should be such that the already implemented strategy, on the time interval $[0,T]$, together with the new optimal strategy, between $T$ and $T'$, form an overall optimal investment strategy on $[0,T']$, as viewed from the initial time. It turns out that the existence of a $U'$ that satisfies this time-consistency property cannot be guaranteed for an arbitrary choice of $U$. Another shortcoming of the classical approach, which is one of the main reasons why it has not become popular among practitioners, is the assumption that the investor's utility function at a (possibly remote) terminal time horizon is known at the initial time. Even though there exist several methods for inferring the investors' preferences from their actions, these methods become less reliable as the time horizon increases.

In order to address the above shortcomings, Henderson \& Hobson and Musiela \& Zariphopoulou, independently, introduced an alternative optimality criterion for the investment problem (cf. \cite{HendersonHobson}, \cite{mz-platen} and \cite{mz-bespoke}).
The associated criterion is developed in terms of a stochastic field, indexed by $T\in\left(0,\infty \right)$ and by the wealth argument $x\in(0,\infty)$, and it is called the \emph{forward investment performance process} (FIPP). The new criterion allows to produce a time-consistent investment strategy that maximizes the expected utility of wealth at every time horizon $T>0$, providing a natural extension of the classical approach.
At the same time, in contrast to the classical framework, the new approach only requires the investor to specify her risk preferences at the very beginning of the trading period and not at a (possibly remote) future time horizon.


\subsection{Forward investment performance process: axiomatic justification}
\label{subse:forwPerf.Ax}

As soon as we deviate from the classical framework and agree that our investment decision should depend on the marginal distributions of the wealth at all times $T>0$, it becomes natural to assume the existence of a family of preferences for the wealth level at every $T>0$. In other words, we assume that, for each $T>0$, there exists a complete order on the space of random variables representing the wealth payoff at time $T$. Assuming, in addition, that these preferences satisfy the usual axioms of Von Neumann and Morgenstern, we conclude that, for each $T>0$, there exists a utility function $U_T$ representing these preferences. Notice, however, that the family of utility functions $\left\{U_T\right\}_{T>0}$ does not represent a complete order on the space of wealth processes. Indeed, for a given pair of wealth processes, the payoff of the first process, at a certain time horizon, may have a higher expected utility than the payoff of the second one, while the opposite relation may hold at a different time horizon. Nevertheless, such a family of preferences may still admit an extremal element -- the wealth process that maximizes all the expected utilities and that can be attained by a strategy which is time-consistent for all time horizons.

Unfortunately, it turns out that there are not many families of classical utility functions that admit an extremal element in the above sense. This is why we have to extend the classical notion of utility function and consider the \emph{state-dependent utilities} (also known as \emph{stochastic utilities}). Notice that the axioms of Von Neumann and Morgenstern are formulated for a space of distributions, and, in particular, the resulting preferences, based on expected utility, only take into account the distribution of the terminal wealth. However, in practice, the investor's preferences often depend upon the joint distribution of the target random variable, say $X_T$, and an additional stochastic factor $Y_T$. For example, the payoff of an investment strategy may be evaluated relative to the inflation factor, or to the overall market performance. If these preferences satisfy the axioms of Von Neumann and Morgenstern (now stated for the pair of random variables $(X_T,Y_T)$), they have to be given by an expected utility, $\EE U(X_T,Y_T)$. Then, the utility function $U(\cdot\,,Y_T)$ is called a state-dependent (or, stochastic) utility. Since the distribution of $Y_T$ is usually specified in the underlying stochastic model (e.g. stochastic volatility), the search for optimal joint distribution of $(X_T,Y_T)$, in fact, reduces to the search for optimal family of conditional distributions of $X_T$, given $Y_T$. Thinking of $Y_T$ as the state, the name of state-dependent utility becomes clear, as it describes the investor's preferences conditional on the state. Using the traditional probabilistic notation, we can also view the state-dependent utility is a random function $U(x,\omega)$, measurable with respect to a given sigma-algebra (generated by $Y_T$). A detailed description of the theory of state-dependent utility can be found in \cite{Dreze}, \cite{Karni1983}, \cite{Karni1985}.

Put simply, the forward investment performance process is a family of state-dependent utilities, indexed by the time horizon $T>0$, and conditioned to admit an optimal investment strategy which maximizes all the expected utilities and which is time-consistent for all time horizons. As mentioned above, such a family of utility functions, typically, does not produce a complete order on the set of available investment strategies (or, the set of attainable wealth processes). It corresponds to the case when the agent does not have preferences over the entire space of strategies (not every two strategies are comparable), but, for any given time horizon $T$ and any state of the relevant market factor $Y_T$, the investor can compare the conditional performance of any two strategies at this time horizon. More precisely, we assume that, for any $T>0$, the investor has a complete order on the space of joint distributions of the wealth process and the relevant stochastic factor, at time $T$, and this order satisfies the axioms of Von Neumann and Morgenstern. Requiring, in addition, the existence of a joint time-consistent optimal strategy for all these preferences, we obtain a forward investment performance process.

\begin{remark}
It is worth mentioning that the concept of recursive utility, introduced in \cite{KrepsPorteus} and \cite{DuffieEpstein}, does not require the axiom of independence and, as a result, produces a very general class of preferences. In particular, the resulting preferences may take into account a wide range of properties of the wealth process -- not only its marginal distributions -- while remaining time-consistent.
However, just like the classical approach, the general recursive utility theory has only been developed for finite time horizons (although some specific constructions for the infinite time horizon are possible). From this point of view, the forward investment performance theory offers something new: its entire purpose is to describe a general class of optimality criteria defined for all positive time horizons, staying as close as possible to the classical theory.
\end{remark}


\subsection{Forward investment performance process: formal definition}\label{subse:forwPerf}

We assume that the market consists of a bank account, whose value, without any loss of generality, stays constant,
and $k$ risky assets $S=\left(S^i,\ldots,S^k\right)$, whose prices are adapted c\`adl\`ag semimartingales on a stochastic basis $\left(\Omega,\mathbb{F}=\left(\mathcal{F}_t\right)_{t\geq0},\mathbb{P}\right)$. All stochastic processes introduced below are defined on this stochastic basis.
As usual, by an investment strategy, or a \emph{portfolio}, we understand a vector $\pi=\left(\pi^1,\ldots,\pi^k\right)^T$ of predictable stochastic processes, integrable with respect to $S$. The investor starts from initial wealth level $x>0$ and allocates her wealth dynamically among the risky securities and the bank account, so that $\pi^i_t$ represents the proportion of her wealth invested in $S^i$ at time $t$. Then, due to the self-financing property, her cumulative wealth process $X^{\pi,x}$ is given by
\begin{equation*}
dX_{t}^{\pi,x}= X_{t}^{\pi,x}\pi^T_t dS_{t},\,\,\,\,\, X^{\pi,x}_0=x,
\end{equation*}
provided $\pi$ is $S$-integrable and locally square integrable.
It is sometimes necessary to consider an even smaller set of portfolios. Hence, we denote by $\mathcal{A}$ the set of \emph{admissible} portfolios, which is a subset of $S$-integrable and locally square integrable processes $\pi$. In addition, we introduce the following notation: $\RR_+ = [0,\infty)$.

\begin{definition}\label{def:FPP}
Given a market model, as above, and a set of admissible portfolios $\mathcal{A}$, a progressively measurable random function $U:\Omega\times\RR_+\times(0,\infty)\rightarrow\RR$ is a forward investment performance process if:


i) Almost surely, for all $t\geq 0$, the mapping $x\rightarrow U_t\left(x\right)$ is concave and increasing;


ii) For any $x>0$ and any $\pi \in \mathcal{A}$, the process $\left(U_t\left( X_{t}^{\pi,x }\right)\right)_{t\geq0}$
is a supermartingale;


iii) For any $x>0$, there exists a portfolio $\pi^{*}\in \mathcal{A}$, such that
$\left(U_t\left( X_{t}^{\pi^{*},x}\right)\right)_{t\geq0}$ is a martingale.
\end{definition}

The property $i)$, in the above definition, simply states that the forward investment performance process is a family of state-dependent utilities, defined for all positive time horizons. The other two properties ensure that this family of utility functions has a unique time-consistent maximizer: an attainable wealth process which maximizes the expected utilities in the given family, for all positive time horizons and initial investment times.

Describing explicitly the space of random functions $U_t(x)$ that satisfy the above
definition is still an open problem, but some results in this direction can be found, for example, in 
\cite{HendersonHobson}, \cite{Berrier}, \cite{ElKaroui}, \cite{ElKaroui.2}, \cite{mz-spde} and \cite{zitkovic}.
In order to present more specific results in this direction, we have to make some additional assumptions on the market model.
In particular, we assume that the filtration $\mathbb{F}$ is
generated by $W$, a standard Brownian motion in $\RR^d$.
In addition, we assume that $S$ is an It\^o process in $\RR^k$ with positive entries, given by
\begin{equation}\label{eq.marketModel.1}
d\log S_{t}=\mu_t dt + \sigma^T_t dW_t,
\end{equation}
where the logarithm is taken entry-wise, $\mu$ is a locally integrable stochastic process with values in $\RR^k$, and $\sigma$ is a $d\times k$ matrix of locally square integrable processes. We use the notation "$A^T$" to denote the transpose of a matrix (vector) $A$.
We introduce the $d$-dimensional stochastic process $\lambda$, frequently called the market price of risk,
via
\begin{equation}\label{lamba.1}
\lambda_t := \left(\sigma^T_t\right)^+ \tmu_t ,
\end{equation}
where $(\sigma^T_t)^+$ is the Moore-Penrose pseudo-inverse of the matrix $\sigma^T_t$, and $\tmu$ is the drift of $S$: $\tmu^i_t=\mu^i_t + \|\sigma^i_t\|^2/2$, for $i=1,\ldots,k$, with $\sigma_t^i$ being the $i$-th column of $\sigma_t$.
In particular, we have
\begin{equation*}
\sigma^T_t \lambda_t = \tmu_t
\end{equation*}
The existence of such a process $\lambda$ follows from the absence of arbitrage in the model.
Notice that, in this case, the cumulative wealth process $X^{\pi,x}$ is given by
\begin{equation*}
dX_{t}^{\pi,x}= X_{t}^{\pi,x}\pi^T_t \sigma^T_t \lambda_t dt + X_{t}^{\pi,x}\pi^T_t \sigma^T_t dW_{t},\,\,\,\,\, X^{\pi,x}_0=x,
\end{equation*}
for any locally square integrable process $\pi$.

Recall that the value function in the classical utility maximization problem, at least formally, solves the Hamilton-Jacobi-Bellman (HJB) equation. It turns out that the following stochastic partial differential equation (SPDE) is an analog of the HJB equation in the forward performance theory:
\begin{equation}\label{SPDE}
dU_t(x) = \frac{1}{2}\frac{\|\partial_xU_t(x)\lambda_t + \sigma_t \sigma^+_t \partial_x a_t(x)\|^2}{\partial^2_xU_t(x)} + a^T_t(x) dW_t,
\end{equation}
where $a_t(x)$ is a $d$-dimensional vector of progressively measurable random functions, continuously differentiable in $x$, which is called a \emph{volatility of the forward performance process}.

Recently, it was shown in \cite{mz-spde}, \cite{zar-RICAM}, and later in \cite{ElKaroui}, \cite{ElKaroui.2}, that, if
$U$ is a twice continuously differentiable stochastic flow (see, for example, \cite{Kunita} for the definition), which satisfies the above SPDE, then,
for any admissible portfolio $\pi$, the process $\left(U_t\left( X_{t}^{\pi,x }\right)\right)_{t\geq0}$ is a local supermartingale (in the sense that there exists a localizing sequence that makes it a supermartingale), and, if, for any initial condition $X^*_0>0$, there exists a strictly positive process $X^*$ satisfying
\begin{equation}\label{wealth-sde}
dX_{t}^* = X_{t}^*\left(\sigma_t\pi^*_t(X^*_t)\right)^T \lambda_t dt + X_{t}^*\left(\sigma_t \pi^*_t(X^*_t)\right)^T dW_{t},
\end{equation}
with
\begin{equation} \label{portfolio-sde}
x \sigma_t \pi _{t}^*(x) = -\frac{\lambda_t \partial_x U_{t}(x) + \sigma_t \sigma^+_t \partial_x a_{t}(x) }{\partial^2_x U_{t}(x) },
\,\,\,\,\,\,\,\,\,\,\,\,\forall x>0,
\end{equation}
then $\left(U_t\left( X_{t}^{*}\right)\right)_{t\geq0}$ is a local martingale.
Of course, according to the definition, the local supermartingale and martingale properties are not sufficient for $U$ to be a forward performance process.
Therefore, having solved the above SPDE (\ref{SPDE}) and constructed the optimal wealth via (\ref{wealth-sde}), one still needs to verify that the resulting process is, indeed, a forward investment performance process (this is analogous to the verification procedure in the classical utility maximization theory). For example, one way to ensure that a local supermartingale $\left(U_t\left( X_{t}^{\pi,x}\right)\right)_{t\geq0}$ is a true supermartingale, is to construct $U$ so that $\inf_{t,x}U_t(x)$ is bounded from below by an integrable random variable. Then, in addition, one can show by a standard argument that the local martingale $\left(U_t\left( X_{t}^{*}\right)\right)_{t\geq0}$ is a true martingale if and only if its expectation at any time coincides with its value at zero.


\subsection{Representation of forward performance processes}

Notice that equation (\ref{SPDE}) may be used to describe the forward performance processes through the volatility $a$. On the other hand, it is not clear what are the admissible choices of volatility -- the ones for which equation (\ref{SPDE}) has a solution. In fact, it is not even clear which ``constant" volatilities (increasing and concave deterministic functions of $x$) are admissible. On the other hand, the results of \cite{ElKaroui.2}, given below, show that there exists a class of volatility processes (although defined in a rather implicit way), for which (\ref{SPDE}) admits a unique solution, for any initial condition satisfying some smoothness and boundedness constraints.
More precisely, it was shown in \cite{ElKaroui.2} that, for any regular enough stochastic flows $\pi^*_t(x)$ and $\nu^*_t(x)$, if the volatility $a$ is specified in the following functional form:
\begin{equation}\label{eq.vol.ElKaroui}
a_t(x) = F\left(t,x,\partial_xU_t(\cdot),\partial^2_xU_{t}(\cdot),\lambda_t,\pi^*_t(\cdot),\nu^*_t(\cdot)\right),
\end{equation}
where $F$ is a given deterministic operator (the same for all choices of $a$), then, there exists a solution to (\ref{SPDE}), for any initial condition $U_0(x)$, which is strictly concave, increasing, satisfies certain smoothness conditions, and takes value zero at $x=0$. In addition, if the resulting solution $U$ is a true forward performance process (i.e. if the local martingale and supermartingale properties are, in fact, global), then the corresponding optimal portfolio is given by $\pi^*$. It is suggested by the authors of \cite{ElKaroui.2} that the above result can be used to solve the problem of \emph{inferring the investor's preferences}. One can, in principle, observe the investor's optimal portfolio $\pi^*$ on some ``test" market and construct the forward performance process $U$ that reproduces this optimal portfolio. Then, naturally, the constructed forward performance process should be used to determine the optimal portfolio in a target market (with different assets and/or a different set of admissible portfolios). However, in a different market, with a different set of attainable wealth processes, the random field $U$ may (and typically does) fail to satisfy the last two properties in Definition $\ref{def:FPP}$ (notice that the definition depends upon the set of available wealth processes). Hence, it fails to produce a time-consistent optimality criterion in the new market.

Even though, at this stage, it is still not clear how to infer investor's preferences using the forward performance theory, the results of \cite{ElKaroui.2} provide analytical representation of a class of forward performance processes. Namely, for a given set of attainable wealth processes $\mathcal{A}$, the forward performance process is described via $\pi^*$, $\nu^*$, and $U_0$. Such a description, definitely, constitutes an important result in the theory of forward performance processes. In particular, it shows that, for any regular enough portfolio process (represented as a random field), there exists a forward performance process that makes the given portfolio optimal. However, from a practical point of view, the assumption that the optimal portfolio $\pi^*$ is known before the optimality criterion is constructed may not always be natural. For example, in the standard optimal investment problem, one uses the optimality criterion in order to construct the optimal portfolio. In addition, the random field $\nu^*$ lacks a clear economic interpretation (although it can be described mathematically, via the dual problem), which makes it difficult to specify its values in particular applications. Therefore, in this paper, we use a different approach to describe the forward performance processes, which is based on the axiomatic justification presented in Subsection $\ref{subse:forwPerf.Ax}$, rather than on the volatility $a$.

Recall that a forward performance process is defined for a given set of attainable strategies $\mathcal{A}$. Therefore, it is natural to think of it as a pair $(U,\mathcal{A})$ that satisfies Definition $\ref{def:FPP}$.
However, in order to give an economic meaning to the forward performance process, one needs to relate it to the investor's preferences on a set of admissible trading strategies. We have accomplished this by identifying a forward performance process with a family of state-dependent utilities. 
A state-dependent utility, in turn, is defined for a given stochastic factor, which causes the state-dependence (or, randomness) of the utility. More precisely, the state-dependent utility represents preferences on conditional distributions, which are constructed by conditioning on the values of the additional stochastic factor. Of course, we {\bf need to define the set of conditional distributions before constructing preferences on it}, or, equivalently, we need to specify the additional stochastic factor before constructing the forward performance process. 
Therefore, in this paper, we propose to identify a forward performance process with a triplet $(U,\mathcal{A},Y)$, where $Y$ represents a stochastic factor that determines the state-dependence of $U$. Namely, we assume that the stochastic field $U$ and the set of attainable claims $\mathcal{A}$ satisfy Definition $\ref{def:FPP}$, and, in addition, $U_t$ is a deterministic function of $(t,x,Y_t)$.
Thus, in order to assign an economic meaning to the forward performance process, we propose that it is defined for a given set of attainable claims $\mathcal{A}$ and for a given stochastic factor $Y$.

Notice that the only novelty of the approach proposed above is in the additional information which is required to identify a forward performance process. 
Namely, the original approach (Definition $\ref{def:FPP}$) requires that the set of attainable claims $\mathcal{A}$ is given, as the additional information needed to identify a forward performance process (i.e. the process is defined for a given $\mathcal{A}$). In the present setting, we require that the stochastic factor $Y_t$, generating the sigma-algebra of $U_t$, is given along with $\mathcal{A}$.
However, in the absence of any assumptions on the stochastic process $Y$, there is no loss of generality in the proposed representation.
To see this, notice that, for any forward performance process $U$, at any time $t$, there exists a random element $Y_t$, such that $U_t$ is a deterministic function of $(t,x,Y_t)$ (e.g. consider the canonical mapping $Y_t: \omega\mapsto\omega$, where only the sigma-algebra of the state space of $Y_t$ changes with $t$). Thus, any possible limitations of the existing framework will arise from the assumptions made on the stochastic factor $Y$, but not from the representation proposed above.

Here, we investigate a regular Markovian case, where the stochastic factor $Y$ is given by a multidimensional diffusion process, and the universe of tradable assets is given by a subset of its components. We say that the associated forward performance process, given by a deterministic function of time, wealth level, and the value of the stochastic factor, is in a \emph{factor form}.
In this case, it turns out that the exact functional relation is determined uniquely by the initial preferences, and, in particular, there is no need to guess the volatility structure of the forward performance process. We characterize the forward performance processes in a factor form via explicit integral representations of the associated positive space-time harmonic functions, and illustrate the theory with specific examples.

The paper is organized as follows. In Subsection $\ref{subse:stochFactor}$, we define the general stochastic factor model, which is a specification of the model described in this section and which remains our framework for the rest of the paper. In Section $\ref{subse:forwardHJB}$ we introduce the forward performance processes in a factor form, as well as the corresponding \emph{time-reversed HJB equation}, and discuss the difficulties associated with it. Sections $\ref{se:linHJB.complete}$ and $\ref{se:linHJB.homothetic}$ demonstrate how, in certain cases, the HJB equation can be reduced to a \emph{backward} linear parabolic equation with \emph{initial} condition.
The main results of this paper are concerned with the representation of positive solutions to the backward linear parabolic equations on the time interval $(0,\infty)$ -- i.e. the \emph{positive space-time harmonic functions}. These results are given in Theorems $\ref{th.strictEllip.rep.abstract}$, $\ref{th:main}$, and $\ref{th:main.2}$ in Section $\ref{se:timeReversedHeat}$. Finally, we consider the closed form examples of forward performance processes in a factor form in Section $\ref{se:examples}$.


\section{Forward performance processes in a factor form}

\subsection{Stochastic factor model}\label{subse:stochFactor}

We assume that the price process of risky assets $S=\left(S^1,\ldots,S^k\right)^T$ is determined by the $n$-dimensional ($n\geq k$) Markov system of stochastic factors $Y=\left(Y^1,\ldots,Y^n\right)^T$. This system is defined on a stochastic basis which supports a $d$-dimensional Brownian motion $B=\left(B^1,\ldots,B^k\right)^T$, via
\begin{equation}\label{eq.marketModel.2}
dY_{t}=\mu(Y_t) dt + \sigma^T(Y_t) dW_t,
\end{equation}
where, with a slight abuse of notation (compare to (\ref{eq.marketModel.1})), we introduce $\mu\in C\left(\RR^n\rightarrow\RR^n\right)$ and $\sigma\in C\left(\RR^n\rightarrow\RR^{d\times n}\right)$, and denote by $\RR^{d\times n}$ the space of $d\times n$ real matrices. We also assume that functions $\mu$ and $\sigma$ are such that the above system has a unique strong solution for any initial condition $y\in\RR^n$. The first $k$ components of $Y$ are interpreted as the logarithms of the tradable securities $S$:
\begin{equation*}
S^i_t = \exp\left(Y^i_t\right),\,\,\,\,\,\,i=1,\ldots,k,
\end{equation*}
and the rest $n-k$ components are the observed, but not tradable, stochastic factors.
In particular, we obtain
\begin{equation*}
dS^i_{t}=S^i_t\tmu^i(Y_t) dt + S^i_t \left(\sigma^i(Y_t)\right)^T dW_t,\,\,\,\,\,\,i=1,\ldots,k,
\end{equation*}
where $\sigma^i(y)$ is the $i$-th column of $\sigma(y)$, and
\begin{equation*}
\tmu^i(y)=\mu^i(y) + \|\sigma^i(y)\|^2/2,\,\,\,\,\,\,\, \forall i=1,\ldots,n
\end{equation*}

Recall that, in this case, the \emph{market price of risk} is given by $\lambda_t=\lambda(Y_t)$, where $\lambda\in C\left(\RR^n\rightarrow\RR^d\right)$ satisfies
\begin{equation}\label{lamba.2}
\left(\sigma^i(Y_t)\right)^T \lambda \left( Y_t\right) = \tmu^i\left(Y_t\right),\,\,\,\,\,\,\forall i=1,\ldots,k
\end{equation}

Given a portfolio $\pi=\left(\pi^1,\ldots,\pi^k\right)^T$, with each $\pi^i$ being a progressively measurable stochastic process with values in $\RR$, we will identify it with the extended $n$-dimensional vector $\left(\pi^1,\ldots,\pi^k,0,\ldots,0\right)^T$ and hope this will not cause any confusion.
Consider an arbitrary dynamic self-financing trading strategy, which starts from initial level $x>0$ and, at each time $t$, prescribes to keep the fraction $\pi^i_t$ of the total wealth invested in $S^i$ (for each $i=1,\ldots,k$). 
Then, the cumulative wealth process of this strategy is given by
\begin{equation*}
dX_{t}^{\pi,x}= X_{t}^{\pi,x}\pi^T_t \tmu(Y_t) dt + X_{t}^{\pi,x}\pi^T_t \sigma^T(Y_t) dW_{t} = X_{t}^{\pi,x}\left(\sigma(Y_t)\pi_t\right)^T \lambda(Y_t) dt
+ X_{t}^{\pi,x}\left(\sigma(Y_t)\pi_t\right)^T dW_{t}
\end{equation*}


\subsection{Time-reversed HJB equation}\label{subse:forwardHJB}

As it was previously announced, we now assume that there exists a function $V:\RR_{+}\times\RR^{n}\times(0,\infty)\rightarrow\RR$, such that the forward performance process $U$ is given in the following \emph{factor form}
\begin{equation}\label{eq.funcForwPerf.def}
U_t\left(x\right) =V\left(t,Y_{t},x\right),
\end{equation}
where $Y$ is defined in (\ref{eq.marketModel.2}).
Our goal is to describe explicitly (in a way which is well suited for implementation) a large class of functions $V$
such that $U$, defined by (\ref{eq.funcForwPerf.def}), is, indeed, a forward performance process.

Assuming enough smoothness, we apply the Ito's formula to $V\left(t,Y_{t},x\right)$ and equate the drift and local martingale terms to those in (\ref{SPDE}). As a result, we obtain the volatility of the forward performance process in a factor form,
\begin{equation*}
a_t\left( x\right) = \sigma(Y_t) D_yV(t,Y_t,x),
\end{equation*}
and derive the following partial differential equation:
\begin{equation}\label{eq.illPosedHJB}
V_t + \max_{\pi\in \RR^k\times\left\{0\right\}^{n-k}} \left[ \left(V_x \lambda +  \sigma D_y V_x \right)^T \sigma \pi +
\frac{1}{2} V_{xx} (\sigma \pi)^T \sigma \pi\right]
+ \frac{1}{2} \text{tr} \left( D_y^2 V \sigma^T \sigma \right) + \left(D_y V\right)^T \mu=0,
\end{equation}
for $(t,y,x)\in(0,\infty)\times\RR^n\times(0,\infty)$. Here, we denote by $D_y V$ the gradient of $V$ (the vector of partial derivatives), and by $D^2_y$ the Hessian of $V$ (the matrix of second order partial derivatives), with respect to $y$.
It is not hard to see that, if $V$ solves the above equation, then, $U_t\left(x\right)=V\left(t,Y_{t},x\right)$ satisfies the last two properties of Definition $\ref{def:FPP}$ locally (that is the `martingale' and `supermartingale' properties are substituted, respectively, to the `local martingale' and `local supermartingale' ones). The proof of the latter statement, as well as the derivation of the above partial differential equation (PDE), are rather standard, hence, we omit the details and, instead, refer the interested reader to \cite{ElKaroui}, \cite{ElKaroui.2}, \cite{zar-RICAM}, and \cite{mz-spde}.

Before we proceed to the construction of solutions to (\ref{eq.illPosedHJB}), it is worth mentioning several important features of the above equation.
First, equation (\ref{eq.illPosedHJB}) provides another way to observe similarities between the forward performance processes and the value functions in the classical utility maximization theory. Indeed, the forward performance process in a factor form satisfies the same equation as the value function, except that it does not have a pre-specified terminal condition at a finite time horizon $T$: instead, the solution is supposed to exist on the entire half line $t>0$. It may seem that the above equation can be reduced to a standard HJB equation by the simple change of variables: $t \mapsto \tau=T-t$, with some fixed $T>0$. However, the resulting (standard HJB) equation can only be solved for $\tau> 0$, and, hence, it produces a solution to (\ref{eq.illPosedHJB}) only for $t\in(0,T)$. This is not sufficient, since the main reason to introduce the forward performance process in the first place was to ensure the time-consistency of the resulting optimization criterion on the \emph{entire half line} $t\in(0,\infty)$. Therefore, unlike the classical HJB equation, (\ref{eq.illPosedHJB}) can only be equipped with initial, rather than terminal, condition and has to be solved \emph{forward} in time. For this reason, we call it a {\bf time-reversed HJB equation}. The requirement that equation (\ref{eq.illPosedHJB}) has to be solved on the entire half-line $t>0$ causes many difficulties in constructing the solutions: on top of all the problems associated with the standard HJB equation (i.e. nonlinearity, degeneracy), the problem at hand has to be solved \emph{in a wrong time direction}, which makes it {\bf ill-posed} from the point of view of the classical PDE theory.

Despite all the difficulties outlined above, we manage to construct solutions to the above equation, under some additional assumptions on the market model. In particular, when the market is complete or the preferences are homothetic in the wealth variable, we characterize explicitly the space of all strictly increasing and concave solutions to the above equation, along with the associated initial conditions, $V(0,\cdot,\cdot)$.


\subsection{Linearizing the HJB equation: complete market case.}\label{se:linHJB.complete}

First, we consider the case of a complete market: i.e. we assume that, at each time $t$, the first $k$ columns of $\sigma(Y_t)$ span the entire $\RR^d$.
Then, the maximization problem inside (\ref{eq.illPosedHJB}) can be solved explicitly, and the HJB equation becomes
\begin{equation}\label{eq.illPosedHJB.complete}
V_t -\frac{1}{2} \frac{\|\lambda V_x + \sigma D_y V_x\|^2}{V_{xx}}
+ \frac{1}{2} \text{tr} \left( D_y^2 V \sigma^T \sigma \right) + D_y V^T \mu=0
\end{equation}
It is well-known that the methods of duality theory allow to linearize the above equation (cf. \cite{KarLehShreve}).
These methods are based on the analysis of the Fenchel-Lagrange dual of $V(t,y,\cdot)$, denoted by $\hat{V}(t,y,\cdot)$. In particular, it is a standard exercise to check that the substitute
\begin{equation}\label{eq.u.complete}
u(t,y,z) = -\hat{V}_x(t,y,\exp(z)) =\left(V_x(t,y,\cdot)\right)^{-1}(\exp(z))
\end{equation}
turns the forward HJB equation (\ref{eq.illPosedHJB.complete}) into the following linear equation:
\begin{equation}\label{eq.illPosedHJB.simplest}
u_{t} + \frac{1}{2}\left[\lambda^T\lambda u_{zz} -  2 D_y u^T_{z} \sigma^T \lambda
+ \text{tr} \left( D_y^2u \sigma^T \sigma \right) \right]
+ \frac{1}{2}\lambda^T\lambda u_{z} +  D_y u^T \left(\mu - \sigma^T \lambda\right) = 0,
\end{equation}
for all $(t,y,z)\in(0,\infty)\times\RR^{n+1}$.
If we manage to find a solution to the above equation and ensure that it is strictly positive and decreasing in $z$, we can then proceed backwards via (\ref{eq.u.complete}), to construct function $V$ that solves (\ref{eq.illPosedHJB.complete}). This step may not always be trivial, as the transition from $V_x$ to $V$ requires integration of the PDE for $V_x$ with respect to $x$. However, this method does work if, for example, we manage to derive sufficient a priori estimates of $u(t,y,z)$ and its partial derivatives, as demonstrated in Subsection $\ref{subse:ex.meanRev}$.


\subsection{Linearizing the HJB equation: homothetic preferences.}\label{se:linHJB.homothetic}

The linearization proposed in the previous subsection relies on the completeness of the market but works for an arbitrary forward performance process in a factor form.
Here, on contrary, we consider the (possibly) incomplete market models, while the forward investment performance process is assumed to be \emph{homothetic} in the wealth argument. Such processes are the natural analogues of the popular power utilities.
More precisely, we assume that, for all $(t,y,x)\in\RR_+\times\RR^n\times(0,\infty)$,
\begin{equation}\label{homotheticity}
V\left( t,y,x\right) =\frac{x^{\gamma}}{\gamma} v\left( t,y\right),
\end{equation}
with some function $v:\RR_+\times\RR^n\rightarrow\RR$ and a non-zero constant $\gamma <1$.
In addition, we make the following specification of the general factor model introduced above.
We assume that $n=d=2$, $k=1$, that $\mu$ and $\sigma$ depend only upon the second component of $y$, and the instantaneous correlation between the two columns of $\sigma$ is constant. In other words, we assume that the market consists of a single risky asset, whose dynamics are given by the following two factor model
\begin{equation*}
\left\{
\begin{array}{l}
{dY^1_t=d\log S_{t}=\mu\left( Y^2_{t}\right)dt + \sigma\left( Y^2_{t}\right) dW_{t}^{1},}\\
{}\\
{dY^2_{t}=b\left( Y^2_{t}\right) dt + a\left( Y^2_{t}\right) \left( \rho dW_{t}^{1}+\sqrt{1-\rho ^{2}}dW_{t}^{2}\right)},
\end{array}
\right.
\end{equation*}
with a constant $\rho\in[-1,1]$ and scalar functions $\mu$, $\sigma$, $a$ and $b$, such that the above system has a unique strong solution for any initial condition $(Y^1_0,Y^2_0)\in\RR^2$.
It is shown in \cite{zar-distortion} that, in the notation
\begin{equation*}
u(t,y):=\left(v(t,y)\right)^{1/\delta},
\end{equation*}
with
\begin{equation*}
\delta=\frac{1-\gamma}{1-\gamma +\rho^2\gamma},
\end{equation*}
the HJB equation (\ref{eq.illPosedHJB}) reduces to
\begin{equation}\label{eq.illPosedHJB.homothetic}
u_{t} + \frac{1}{2}a^{2}\left( y\right) u_{yy} + \left( b\left( y\right)
+ \rho \frac{\gamma }{1-\gamma }\lambda \left( y\right) a\left( y\right) \right) u_{y}
 + \frac{1}{2\delta}\frac{\gamma}{1-\gamma }\lambda ^{2}\left( y\right)u=0,
\end{equation}
for all $(t,y)\in(0,\infty)\times\RR^{n}$, where $\lambda(y)=\mu(y)/\sigma(y) + \sigma(y)/2$.
Thus, we have reduced the time-reversed HJB equation (\ref{eq.illPosedHJB}) to a linear parabolic equation. Solving the above equation, we obtain function $u(t,y)$ and, taking its power, recover $v$ and, in turn, $V$.

Notice however, that the above equation, as well as (\ref{eq.illPosedHJB.simplest}), is time-reversed: it has to be solved forward, for $t\in(0,\infty)$, while the associated differential operator corresponds to a backward equation. We would like to emphasize that there is no standard existence theory for such PDEs. Developing some basic existence results for this type of equations is the subject of the next section.


\section{Generalized Widder's theorem as the representation of space-time harmonic functions}\label{se:timeReversedHeat}

In this section, we show how to generate solutions to a class of time-reversed (ill-posed) linear parabolic equations on a semi-finite time interval, which includes (\ref{eq.illPosedHJB.simplest}) and (\ref{eq.illPosedHJB.homothetic}).
These results, in particular, provide an extension of the Widder's theorem on positive solutions to the heat equation (see \cite{Widder1963}). We recall this theorem and provide additional comments further in this section.


\subsection{Uniformly parabolic case}

Here, we consider linear parabolic equations of the form
\begin{equation}\label{eq.reversedHeatEq}
u_t + \mathcal{L}_y u = 0,\,\,\,\,\,\,\,\,(t,y)\in (0,\infty)\times\RR^n,
\end{equation}
with the operator $\mathcal{L}_y$ given by
\begin{equation}\label{eq.Ly.def}
\mathcal{L}_y= \sum_{i,j=1}^n a^{ij}(y) \partial^2_{y^i y^j} + \sum_{i=1}^n b^i(y) \partial_{y^i} + c(y),
\end{equation}
where the functions $a^{ij}$, $b^i$ and $c$ are uniformly H\"older-continuous and absolutely bounded, and such that the matrix
$A=(a^{ij})$ is symmetric and satisfies the \emph{uniform ellipticity} condition:
\begin{equation}\label{eq.strictEllip}
0 < \inf_{\|v\|=1,\,y\in\RR^n} \sum_{i,j=1}^n v_i v_j a^{ij}(y)
\end{equation}
The operator $\mathcal{L}_y$ is, then, called \emph{uniformly elliptic}, and the equation (\ref{eq.reversedHeatEq}) is \emph{uniformly parabolic}.
Notice that (\ref{eq.reversedHeatEq}) can be rewritten as the evolution equation
\begin{equation*}
u_t = -\mathcal{L}_y u,
\end{equation*}
where `$-\mathcal{L}_y$' is an ``anti-elliptic" (positive) operator. According to the classical theory of linear parabolic equations (see, for example, \cite{Evans}), in order to solve the above equation forward in time (with a given initial condition), one needs the operator in the right hand side to be elliptic (negative), and, hence, it cannot be applied in this case. In fact, as we will show later, it is not always possible to construct a solution to the above equation, even for a smooth initial condition satisfying the usual growth constraints (or, having a compact support). Nevertheless, we will provide an explicit description of the space of all initial conditions for which the nonnegative solution to (\ref{eq.reversedHeatEq}) does exist.

To begin, consider the simplest possible form of equation (\ref{eq.reversedHeatEq}):
\begin{equation}\label{eq.revHeat}
u_t + u_{yy} = 0,\,\,\,\,\,\,\,\,(t,y)\in (0,\infty)\times\RR
\end{equation}
As mentioned earlier, the nonnegative solutions of the above equation are completely
characterized by the celebrated Widder's theorem, given below (see Theorem 8.1 in \cite{Widder1963}).


\begin{theorem}\label{th:widder}
(Widder 1963) Function $u:(0,\infty)\times \RR \rightarrow \RR$ is a positive classical solution to (\ref{eq.revHeat}) if and only if it can be represented as
\begin{equation}\label{eq.repres.Widder}
u\left( t,y\right) = \int_{\RR} e^{zy-z^{2}t} \nu\left( dz\right)
\end{equation}
where $\nu $ is a Borel measure, such that the above integral is finite for all $(t,y)\in(0,\infty)\times\RR$.
\end{theorem}


As the above theorem shows, the \emph{only} functions that can serve as
initial conditions to (\ref{eq.revHeat}) are given by the bilateral Laplace
transform of the underlying measure $\nu$, namely,
\begin{equation*}
u\left(0,y\right) = \int_{\RR} e^{yz} \nu\left( dz\right),
\end{equation*}
provided the above integral converges for any $y\in \RR$.
We can, now, see that there exists a non-empty space of positive (nonnegative) solutions to equation (\ref{eq.revHeat}), which, of course, is a convex cone. This space is different from the spaces we usually consider when constructing the solutions to a standard elliptic or parabolic linear equation. In particular, as follows from the above representation, one cannot expect the solutions of (\ref{eq.revHeat}) to be vanishing at $y\rightarrow\infty$ and $y\rightarrow-\infty$ simultaneously.
It is also easy to see, by choosing the measure $\nu$ with atoms at the nonnegative integers $\left\{n\right\}$, with the corresponding weights $\left\{1/n!\right\}$, that there exists a solution of (\ref{eq.revHeat}) with the initial condition
\begin{equation*}
u\left(0,y\right) = \int_{\RR} e^{yz} \nu\left( dz\right)=\exp\left(e^y\right)
\end{equation*}
Recall that the above function does not satisfy the necessary growth restriction, and, hence, the standard heat equation
\begin{equation*}
u_t - u_{yy} = 0,\,\,\,\,\,\,\,\,(t,y)\in (0,\infty)\times\RR,
\end{equation*}
equipped with this initial condition, does not possess a solution. Thus, one cannot claim that the space of solutions to (\ref{eq.revHeat}) is ``smaller" than the space of solutions to the standard heat equation. Rather, it is a different space of functions which do not posses some of the properties that we are used to consider natural.

Widder's theorem was used in \cite{HendersonHobson}, \cite{Berrier} and \cite{mz-spde} to describe a class of forward performance processes with zero volatility, which are not necessarily in the factor form proposed herein. Recall that, here, we focus on describing the forward performance processes in a factor form, which may have a nontrivial (i.e. non-zero) volatility.
In particular, the goal of this subsection is to describe the space of solutions to the general time-reversed uniformly parabolic equation (\ref{eq.reversedHeatEq}).
The techniques used by Widder to prove the representation (\ref{eq.repres.Widder}) are based on applying a specific function transform in the space variable and cannot be extended easily to the general case. Therefore, we have to develop a new method for studying equation (\ref{eq.reversedHeatEq}) in full generality.

In fact, the solutions to (\ref{eq.reversedHeatEq}) are called the \emph{space-time harmonic functions} associated with the operator ``$\partial_t + \mathcal{L}_y$". From the probabilistic point of view, these functions characterize the \emph{Martin boundary} of a \emph{space-time diffusion process} $(t,y_t)$, where $\left(y_t\right)$ is the diffusion associated with the generator $\mathcal{L}_y$. For the precise definitions of Martin boundary and its relation to harmonic functions, we refer to \cite{Doob}, \cite{Pinsky}, \cite{RogersWilliams}.
It turns out that one can obtain an explicit integral representation of all space-time harmonic functions using the methods of Potential Theory. These methods allow to describe the Martin boundary of a space-time diffusion via the Martin boundary of the space process itself, which, from an analytical point of view, reduces the ill-posed equation (\ref{eq.reversedHeatEq}) to a well-posed uniformly elliptic equation.
In particular, the results presented below are based on the representation of the \emph{minimal} elements of the cone of nonnegative space-time harmonic functions, obtained by Koranyi and Taylor in \cite{KoranyiTaylor}. The application of Choquet's theory, then, allows us to derive a representation of all solutions to (\ref{eq.reversedHeatEq}) via the minimal solutions, which, in turn, can be computed by solving the associated (well-posed) elliptic equations. This result, in particular, provides a generalization of the Widder's theorem stated above.
However, in order to apply the results of Koranyi and Taylor to the problem at hand, we need to make some additional constructions.


\begin{definition}\label{def.V}
The space $\mathcal{V}$ consists of all functions $v:\left((0,\infty)\times\RR^n\right)\cup\left\{(0,0)\right\}\rightarrow \RR$, continuous on any set $M_{\alpha}:=\left\{\left.(t,y)\in[0,\infty)\times\RR^n\,\right|\,t\geq \alpha \|y\|^2\right\}$, for any $\alpha>0$. The set $\mathcal{V}$ is endowed with the topology of uniform convergence on any compact contained in some $M_{\alpha}$.
\end{definition}


\begin{definition}
The space $\mathcal{H}$ consists of all functions $u\in\mathcal{V}$, such that: $u\in C^{1,2}\left((0,\infty)\times\RR^n\right)$, $u\geq0$, $u(0,0)=1$, and $u$ satisfies (\ref{eq.reversedHeatEq}).
\end{definition}


\begin{definition}
Function $u\in \mathcal{H}$ is a minimal element of $\mathcal{H}$ if, for any $v\in \mathcal{H}$, $v\leq u$ implies $v=\lambda u$, for some $\lambda\in[0,1]$.
\end{definition}

The main result of \cite{KoranyiTaylor} provides an explicit characterization of the minimal elements of $\mathcal{H}$ (i.e. the minimal positive solutions to (\ref{eq.reversedHeatEq})).

\begin{definition}
The set $\mathcal{E}$ consists of all functions $v:\left((0,\infty)\times\RR^n\right)\cup\left\{(0,0)\right\}\rightarrow\RR$ of the form $v(t,y) = e^{-\lambda t} \psi(y)$, with any $\lambda\in\RR$ and any $\psi\in C^2(\RR^n)$, such that $\psi(0)=1$, $\psi\geq0$, and $(\mathcal{L}_y-\lambda)\psi(y)=0$ for all $y\in\RR^n$.
\end{definition}


\begin{theorem}\label{th:KoranyiTaylor}
(Koranyi-Taylor, 1985) The set of all minimal elements of $\mathcal{H}$ coincides with $\mathcal{E}$.
\end{theorem}
\begin{proof}
The proof is given in \cite{KoranyiTaylor} and it is based on the uniform Harnack's inequality for the solutions of (\ref{eq.reversedHeatEq}). See Appendix A for a relevant version of Harnack's inequality.
\end{proof}


In fact, Koranyi and Taylor show that $\mathcal{E}$ is the set of all minimal elements of a larger space of solutions. Notice that, in the definition of $\mathcal{V}$, we restricted the space of functions to those that are continuous on the parabolic shapes centered at zero. However, it is clear that all elements of $\mathcal{E}$ belong to $\mathcal{H}$, which, combined with the results of \cite{KoranyiTaylor}, yields the statement of the above theorem. The reason that we restrict our analysis to the space $\mathcal{H}$ is that, in order to provide an explicit representation of all elements of $\mathcal{H}$, we need this space to be compact in a topology which makes delta-function a continuous functional. The space proposed by Koranyi and Taylor does not satisfy this property, which is, perhaps, the reason why the aforementioned representation was not established in \cite{KoranyiTaylor}. Notice that $\mathcal{H}$ includes all solutions to (\ref{eq.reversedHeatEq}) which are continuous at $t=0$ and, hence, from an application point of view, our restriction is no loss if generality. 


\begin{lemma}\label{le:H.compact}
The set $\mathcal{H}\subset \mathcal{V}$ is compact.
\end{lemma}

\begin{proof}

This result follows from Harnack's inequality and Schauder estimates (see Appendix A).

It is clear that the topology of $\mathcal{V}$ (and, respectively, of $\mathcal{H}$) is equivalent to the topology of uniform convergence on the sets
$$
M^R_{\alpha}:=M_{\alpha}\cap B_R(0,0),
$$
for all $\alpha,R>0$,
where $B_R(0,0)$ is the ball of radius $R$ in $\RR^{1+n}$, centered at zero.
The Harnack's inequality (see Appendix A) implies that, for any $R>0$, there exists a constant $C(R)$, depending only on $R$, on the upper bounds of the absolute values of the coefficients in $\mathcal{L}_y$, and on the lower and upper bounds of the associated quadratic form, such that any nonnegative solution $u$ of equation (\ref{eq.reversedHeatEq}) satisfies:
\begin{equation*}
u(R,y) \leq C(R) u(0,0)=C(R),\,\,\,\,\,\,\,\,\,\forall \|y\|^2\leq1
\end{equation*}
For any $\lambda\in(0,1)$ and $r>0$, we introduce function $v^{\lambda}(t,y):=u(\lambda^2 t,y\lambda\sqrt{r})$ and notice that it satisfies a strictly parabolic PDE whose coefficients and the associated quadratic form can be bounded by a function of $r$, uniformly over $\lambda\in(0,1)$.
Therefore, there exists a constant $C'(\alpha,R)>0$, such that
\begin{equation*}
u(R\lambda^2,y)=v^{\lambda}\left(R,\frac{y}{\lambda\sqrt{r}}\right) \leq C'(r,R),\,\,\,\,\,\,\,\,\,\forall \|y\|^2\leq r\lambda^2,\,\,\,\forall \lambda\in(0,1)
\end{equation*}
This implies that all elements of $\mathcal{H}$ are bounded uniformly on each $M^R_{\alpha}$, with $\alpha = R/r$.
This conclusion, together with the interior Schauder estimates (see Theorem $1$ in \cite{Schauder}, or Appendix A), yield the relative compactness of $\left\{\left.u\,\right|\,u\in\mathcal{H}\right\}$, $\left\{\left.\mathcal{L}_yu\,\right|\,u\in\mathcal{H}\right\}$, and $\left\{\left.u_t\,\right|\,u\in\mathcal{H}\right\}$ as the subsets of $\mathcal{V}$. Thus, we conclude that any sequence in $\mathcal{H}$ has a convergent subsequence whose limit belongs to $\mathcal{H}$. Since the topology in $\mathcal{V}$ is metrizable, this completes the proof of the lemma.
\end{proof}

Before we can formulate the main theorems, we need to recall some auxiliary results.


\begin{definition}
A function $u\in \mathcal{H}$ is an extreme element of $\mathcal{H}$ if, for any $v_1, v_2 \in \mathcal{H}$, $\frac{1}{2} v_1 + \frac{1}{2} v_2 = u$ implies $v_1=v_2=u$.
\end{definition}


\begin{lemma}
The set of extreme points of $\mathcal{H}$ coincides with the set of its minimal elements $\mathcal{E}$.
\end{lemma}
\begin{proof}
This is a standard result from Potential Theory (cf. page $33$ of \cite{Doob}).
\end{proof}


\begin{lemma}\label{le:E.mbl}
The set $\mathcal{E}\subset \mathcal{V}$ is Borel.
\end{lemma}

\begin{proof}
This is a standard result from Convex Analysis (see Proposition $1.3$ in \cite{Phelps}).
\end{proof}


The following theorem is an immediate corollary of the above results.


\begin{theorem}\label{th.strictEllip.rep.abstract}
Function $u$ belongs to $\mathcal{H}$ (is a nonnegative classical solution to (\ref{eq.reversedHeatEq}), normalized at zero) if and only if there exists a Borel probability measure $\nu$ on $\mathcal{E}$, such that, for any $(t,y)\in \left((0,\infty)\times\RR^n\right)\cup\left\{(0,0)\right\}$, we have
\begin{equation}\label{eq.repres.abstract}
u(t,y) = \int_{\mathcal{E}} v(t,y) \nu(dv)
\end{equation}
Such a measure $\nu$ is uniquely determined by $u\in \mathcal{H}$.
\end{theorem}
\begin{proof}
In view of Lemma $\ref{le:H.compact}$, the necessity of this statement follows immediately from the Choquet's theorem (cf. page $14$ of \cite{Phelps}), and the sufficiency is a well known result from convex analysis (see Proposition $1.1$ in \cite{Phelps}).
\end{proof}


The above theorem is nothing else but a version of the abstract \emph{Martin representation theorem} (cf. Chapter XII.9 in \cite{Doob}), with the only exception that, here, we are able to describe the topology of $\mathcal{E}$ explicitly. However, the structure of the Borel measures on $\mathcal{E}$ is, still, not very clear, making it difficult to apply the above representation in practice.
Therefore, below, we formulate another result, which is equivalent to Theorem $\ref{th.strictEllip.rep.abstract}$, but is better suited for computations (as demonstrated in Section $\ref{se:examples}$).


\begin{theorem}\label{th:main}
Function $u$ belongs to $\mathcal{H}$ (is a nonnegative classical solution to (\ref{eq.reversedHeatEq}), normalized at zero) if and only if it can be represented, for all $(t,y)\in\left((0,\infty)\times\RR^n\right)\cup\left\{(0,0)\right\}$, as
\begin{equation}\label{eq.represent}
u(t,y) = \int_{\RR} e^{-t\lambda} \psi(\lambda;y) \mu(d\lambda),
\end{equation}
with a Borel probability measure $\mu$ on $\RR$ and a nonnegative function
$\psi: \RR\rightarrow C^2(\RR^n)$, such that $\psi\in L^{1}\left(\RR\rightarrow C(\mathcal{K});\mu\right)$
for any compact $\mathcal{K} \subset \RR^n$ and, for $\mu$-almost every $\lambda$, the following holds: $\psi(\lambda,0)=1$ and $\psi(\lambda;\cdot)$ solves
\begin{equation}\label{eq.elliptic}
\left(\mathcal{L}_y - \lambda\right) \psi(\lambda;y)=0,
\end{equation}
for all $y\in\RR^n$.
Such a pair $(\mu,\psi)$ is determined uniquely by $u\in \mathcal{H}$.

\end{theorem}

\begin{remark}
The main contribution of Theorem $\ref{th:main}$ is that it reduces the (ill-posed) forward parabolic equation (\ref{eq.reversedHeatEq}), which cannot be analyzed by means of standard theory, to a regular elliptic equation (\ref{eq.elliptic}), which can be solved using the existing methods. 
In particular, if $n=1$, all positive solutions to the one-dimensional version of (\ref{eq.elliptic}) can be described through the two (increasing and decreasing) fundamental solutions, which, in turn, can be approximated efficiently, for example, by a series expansion (cf. \cite{Titchmarsh}).
Some existence results for an arbitrary dimension $n$ are also presented in Appendix A. 
\end{remark}

\begin{proof}
Let's prove the necessity first. We need to derive the representation (\ref{eq.represent}) from (\ref{eq.repres.abstract}).
Consider $\mathcal{E}$ as a random space, with the Borel sigma-algebra (the topology is induced by $\mathcal{V}$) and a probability measure $\nu$ on it.
Recall that each $v\in \mathcal{E}$ has a unique decomposition: $v(t,y)=e^{-\lambda t} \psi(y)$.
Then, we fix an arbitrary $\varepsilon\in(0,1)$ and a compact $\mathcal{K}\subset \RR^n$, and introduce the following random elements:
$$
\xi: \mathcal{E} \rightarrow C\left([\varepsilon,1/\varepsilon]\right),
\,\,\,\,\,\,\,v\mapsto v(\cdot,0),
$$
$$
\eta: \mathcal{E} \rightarrow C\left(\mathcal{K}\right),
\,\,\,\,\,\,\,v\mapsto v(0,\cdot),
$$
$$
\zeta: \mathcal{E} \rightarrow \RR,
\,\,\,\,\,\,\,v \mapsto \log\left(\xi(v)(1)\right),
$$
where the `$C$' spaces are endowed with uniform norms, making them into Banach spaces.
The above mappings are continuous and, hence, measurable. In addition, a simple application of Harnack's inequality (see, for example, the proof of Lemma $\ref{le:H.compact}$) shows that the respective norms of $\xi(v)$, $\eta(v)$, and $\zeta(v)$ are bounded over all $v\in\mathcal{E}$.
Next, notice that, for any $(t,y)\in [\varepsilon,1/\varepsilon]\times\mathcal{K}$, we have
$$
\int_{\mathcal{E}} v(t,y)\, \nu(dv) = \left[\int_{\mathcal{E}} \xi(v)\eta(v)\, \nu(dv)\right](t,y) = \left[\EE (\xi\eta)\right](t,y)
=\left[\EE\left(\xi\,\EE\left[\left. \eta\,\right|\, \zeta \right]\right)\right](t,y),
$$
where the second integral is understood in the Bochner sense (see Appendix A for details), and, to obtain the last equality, we noticed that the value of $\xi(v)$ is uniquely determined by the value of $\zeta(v)$.
The argument $(t,y)$ can be put in and out of the second integral in the above, due to the fact that delta-function is a continuous functional with respect to the uniform topology and due to the properties of Bochner integral (see the Hille's theorem in Appendix A or in \cite{SwartzFuncAn}).
Next, recall the basic property of conditional expectation, which states that there exists
$\psi\in L^{1}\left(\RR\rightarrow C(\mathcal{K});\mu\right)$, with $\mu$ being the distribution of $\zeta: \mathcal{E} \rightarrow \RR$,
such that $\EE\left[\left. \eta\right| \zeta \right] = \psi(\zeta)$.
Therefore, we have
$$
\int_{\mathcal{E}} v(t,y) \nu(dv) = \left[\EE\left(\xi\, \psi(\zeta) \right)\right](t,y)
=  \left[\int_{\mathcal{E}} \xi(v) \psi(\zeta(v)) \nu(dv)\right](t,y)
$$
$$
= \int_{\mathcal{E}} e^{-t\zeta(v)} \psi(\zeta(v);y) \nu(dv)
= \int_{\RR} e^{-t\lambda} \psi(\lambda;y) \mu(d\lambda)
$$
The integral in the right hand side of the above ia absolutely convergent, as such is the integral in the left hand side.
Thus, we obtain the desired representation (\ref{eq.represent}).

To prove that function $u$ defined by (\ref{eq.represent}) belongs to $\mathcal{H}$, we, first, recall the well known fact (see, for example, Theorem 4.3.2 in \cite{Pinsky}) that there exists $\lambda_0\in\RR$, such that for any $\lambda<\lambda_0$ the only nonnegative solution to (\ref{eq.elliptic}) is zero. Thus, the support of $\mu$ is bounded from below, and, hence, the integral in (\ref{eq.represent}) is well defined.
Next, we notice that the mapping
$$
\RR\ni \lambda \mapsto \left((t,y) \mapsto e^{-t\lambda} \psi(\lambda;y)\right) \in \mathcal{E}
$$
is measurable and, hence, we can use a change of variables to deduce
$$
u(t,y)=\int_{\RR} e^{-t\lambda} \psi(\lambda;y) \mu(d\lambda)=\int_{\mathcal{E}} v(t,y) \nu(dv),
$$
for some probability measure $\nu$ on $\mathcal{E}$ and any $(t,y)\in (0,\infty)\times\RR^n$.
We now apply the standard result from convex analysis (cf. Proposition 1.1 in \cite{Phelps}), which states that an integral with respect to a probability measure over a compact convex set in a locally convex space \emph{represents} a point in this set (in the sense that the value of any continuous linear functional applied to this point coincides with the integral of the values of this functional applied to the integrand). In the present case, it means that $u\in \mathcal{H}$.

Let's prove the uniqueness of such representation. Assume there exists another pair $(\mu',\psi')$ such that
$$
u(t,y) = \int_{\RR} e^{-t\lambda} \psi'(\lambda;y) \mu'(d\lambda).
$$
Consider $\mu''=\frac{1}{2}(\mu+\mu')$. It is a probability measure, and we have: $\mu\prec \mu''$ and $\mu'\prec \mu''$. Denote the densities of $\mu$ and $\mu'$, with respect to $\mu''$, by $p$ and $p'$ respectively. Notice that, for $\mu''$-almost every $\lambda$, we have $\psi(\lambda;0)=\psi'(\lambda;0)=1$. Thus, we obtain
$$
u(t,0)=\int_{\RR} e^{-t\lambda} p(\lambda) \mu''(d\lambda)
= \int_{\RR} e^{-t\lambda} p'(\lambda) \mu''(d\lambda)
$$
for all $t\geq0$.
Recall that the supports of $\mu$ and $\mu'$ have to lie in $[\lambda_0,\infty)$, for some $\lambda_0\in\RR$. Therefore, we obtain
$$
\int_{\lambda_0}^{\infty} e^{-t\lambda} p(\lambda) \mu''(d\lambda)
= \int_{\lambda_0}^{\infty} e^{-t\lambda} p'(\lambda) \mu''(d\lambda)
$$
From the uniqueness of the integral representation in the Bernstein (or, Widder-Arendt) theorem (cf. Theorem II.6.3 in \cite{WidderLaplace}), we conclude that $p\equiv p'$, and, hence, $\mu\equiv\mu'$.
As a result, we have
$$
\int_{\lambda_0} e^{-t\lambda} \psi(\lambda;y) \mu(d\lambda)
= \int_{\lambda_0} e^{-t\lambda} \psi'(\lambda;y) \mu(d\lambda).
$$
Finally, we apply the generalized Widder-Arendt theorem (see Theorem 1.2 in \cite{Chojnacki}), to conclude that $\psi$ and $\psi'$ coincide, as elements of $L^{1}\left(\RR\rightarrow C(\mathcal{K});\mu\right)$.
\end{proof}


We finish this subsection by recovering the Widder's representation (\ref{eq.repres.Widder}) from Theorem $\ref{th:main}$.
Recall that, if $\mathcal{L}_y=\Delta$ and $n=1$, any solution to (\ref{eq.elliptic}) is a linear combination of the following fundamental solutions
\begin{equation*}
\psi^{1}(y,\lambda )=e^{y\sqrt{\lambda }}\,\,\,\,\,\,\text{and}
\,\,\,\,\,\,\psi^{2}(y,\lambda) = e^{-y\sqrt{\lambda }},
\end{equation*}
for all $\lambda\geq0$. And there are no positive solutions to (\ref{eq.elliptic}) if $\lambda<0$.
Thus, according to Theorem $\ref{th:main}$, all nonnegative solutions to (\ref{eq.reversedHeatEq}) are given by
\begin{equation*}
u(t,y) = \int_0^{\infty}e^{-\lambda t}\left(c_1(\lambda) e^{-y\sqrt{\lambda }} + c_2(\lambda) e^{y\sqrt{\lambda }}\right)\nu(d\lambda),
\end{equation*}
where $\nu$ is a Borel measure, and $c_i$'s are measurable nonnegative functions, such that the above integral converges everywhere.
Changing variables in the above, we obtain the Widder's representation:
\begin{eqnarray*}
&&u(t,y) = \int_{\RR} e^{yz-z^2 t} \left(\nu_1(dz)+ \nu_2(dz)\right),
\end{eqnarray*}
where
\begin{equation*}
\nu_{1}(dz) = \bone_{(-\infty,0]}(z) c_1(z^2)\left(\nu\circ m^{-1}_1\right)(dz)\,\,\,\,\,\,\,\text{and}\,\,\,\,\,\,\nu_{2}(dz) = \bone_{[0,\infty)}(z) c_2(z^2)\left(\nu\circ m^{-1}_2\right)(dz),
\end{equation*}
with $m_1:\lambda\mapsto-\sqrt{\lambda}$ and $m_2:\lambda\mapsto\sqrt{\lambda}$.

\begin{remark}
It is worth discussing the connection between the representation (\ref{eq.represent}) and the \emph{turnpike} theorems, developed, for example, in \cite{Mossin}, \cite{CoxHuangTurnpike}, \cite{DetempleTurnpike}, \cite{GuasoniRobertson}. These papers consider solutions to a sequence of optimal investment problems, with the same utility function and the time horizons going to infinity. Assuming that the optimal wealth processes, for all the optimization problems, are bounded from below by a deterministic process exploding at infinity, and, in addition, that the utility function behaves like a power function, asymptotically, for large wealth arguments, the turnpike theorems yield
\begin{equation*}
u(t,y) \sim e^{-\lambda t} \psi(\lambda;y),
\end{equation*}
as the time horizon $t$ grows to infinity. Function $u$, in this case, is understood as the inverse of the marginal value function of a finite time horizon problem. Notice that our results are in perfect accordance with the turnpike theorems: Theorem $\ref{th:main}$ implies that, as the time horizon goes to infinity, the asymptotic relation of the turnpike theorems holds for a sequence of problems with state- and time-dependent utility functions, which have power dependence on the wealth argument. However, unlike the turnpike theorems, here, we consider only time-consistent sequences of optimization problems, which have a common solution for all time horizons, and we obtain an exact, rather than asymptotic, relation.
\end{remark}


\subsection{Degenerate case}

Notice that not all equations arising in the portfolio optimization theory are of the form (\ref{eq.reversedHeatEq}).
In fact, as it was demonstrated in Subsection $\ref{se:linHJB.complete}$, in complete diffusion-based markets, the application of duality methods typically leads to the following equation:
\begin{equation}\label{eq.reversedHeatEq.mod}
u_t + \mathcal{L}_{yz} u = 0,\,\,\,\,\,(t,y,z)\in (0,\infty)\times\RR^{n+1},
\end{equation}
where
\begin{equation*}
\mathcal{L}_{yz}= \sum_{i,j=1}^n a^{ij}(y) \partial^2_{y^i y^j}
+ \sum_{i=1}^n q^i(y) \partial^2_{zy^i} + p(y) \partial^2_{zz}
+ \sum_{i=1}^n b^i(y) \partial_{y^i} + r(y) \partial_z + c(y),
\end{equation*}
with continuous functions $\left\{a^{ij}\right\}$, $p$, $\left\{q^i\right\}$, $\left\{b^i\right\}$, $r$, and $c$, defined through the parameters of the stochastic model:
\begin{eqnarray*}
&&\left(a^{ij}(y)\right) = \sigma^T(y) \sigma(y),\,\,\,\,\,\,\,\,
q(y) = \sigma^T(y) \lambda(y),\,\,\,\,\,\,\,\,\,\,\,\,
p(y) = \lambda^T(y) \lambda(y),\\
&&b(y) = \mu(y) - \sigma^T(y) \lambda(y),\,\,\,\,\,\,\,\,\,
r(y) = \frac{1}{2}\lambda^T(y)\lambda(y),\,\,\,\,\,\,\,\,\,
c(y) = 0.
\end{eqnarray*}
One can see that the quadratic form of $x\in\RR^{n+1}$, associated with $\mathcal{L}_{yz}$,
$$
\sum_{i,j=1}^n a^{ij}(y) x^i x^j + \sum_{i=1}^n q^i(y) x^i x^{n+1} + p(y) (x^{n+1})^2,
$$
is degenerate in, at least, one direction, at each point $y\in\RR^{n}$, implying that $\mathcal{L}_{yz}$ is \emph{not} uniformly elliptic (but rather \emph{degenerate elliptic}), as an operator acting on functions on $\RR^{n+1}$.
As a consequence, many of the techniques used in the previous subsection (in particular, the uniform Harnack's inequality), cannot be applied to equation (\ref{eq.reversedHeatEq.mod}).
To illustrate the differences, we follow the ideas of previous subsection and introduce the space $\tilde{\mathcal{E}}$.


\begin{definition}
The set $\tilde{\mathcal{E}}$ consists of all functions $v:\left((0,\infty)\times\RR^{n+1}\right)\cup\left\{(0,0,0)\right\}\rightarrow\RR$ of the form $v(t,y,z)= e^{-\lambda t} \psi(y,z)$, with any $\lambda\in\RR$ and any $\psi\in C^2(\RR^{n+1})$, such that $\psi(0,0)=1$, $\psi\geq0$, and $(\mathcal{L}_{yz}-\lambda)\psi(y,z)=0$ for all $(y,z)\in\RR^{n+1}$.
\end{definition}


We endow $\tilde{\mathcal{E}}$ with the topology of uniform convergence on any compact contained in
\begin{equation}
\tilde{M}_{\alpha}:=\left\{\left.(t,y,z)\in[0,\infty)\times\RR^{n+1}\,\right|\,t\geq \alpha \left(\|y\|^2+z^2\right)\right\},
\end{equation}
for any $\alpha>0$.
It is, then, natural to suggest that all nonnegative solutions to (\ref{eq.reversedHeatEq.mod}), normalized at zero, are given by
\begin{equation}\label{eq.repres.abstract.deg}
u(t,y,z) = \int_{\tilde{\mathcal{E}}} v(t,y,z) \nu(dv)
\end{equation}
for all $(t,y,z)\in \left((0,\infty)\times\RR^{n+1}\right)\cup\left\{(0,0,0)\right\}$,
where $\nu$ is a Borel probability measure on $\tilde{\mathcal{E}}$.
However, it turns out that the above representation is not complete!


Let us construct an example of equation of the type (\ref{eq.reversedHeatEq.mod}), which possesses a solution that cannot be represented in the form (\ref{eq.repres.abstract.deg}).
Consider the simplest case when our model reduces to the one-dimensional Black-Scholes-Merton model, with
\begin{eqnarray*}
&&n=1;\,\,\,\,\sigma(y) = \sigma \in (0,\infty);\,\,\,\,\,\,\,\,
\mu(y) = \tilde{\mu} - \sigma^2/2,\,\,\,\text{with}\,\,\, \tilde{\mu}\in \RR;\,\,\,\,\,\,
\lambda(y) = \frac{\tilde{\mu}}{\sigma}\in\RR
\end{eqnarray*}
The equation (\ref{eq.reversedHeatEq.mod}), then, reduces to
\begin{equation}\label{eq.reversedHeatEq.mod.BS}
u_t + \frac{\sigma^2}{2}\left(u_{yy} - 2\frac{\lambda}{\sigma}u_{zy} + \frac{\lambda^2}{\sigma^2}u_{zz}\right) + \frac{\lambda^2}{2}u_z - \frac{\sigma^2}{2}u_y = 0,\,\,\,\,\,(t,y,z)\in (0,\infty)\times\RR^{2}
\end{equation}
Assuming $\tilde{\mu}\neq\sigma^2$ and $\tilde{\mu}\neq0$, we choose a smooth function $\varphi:\RR\rightarrow[0,\infty)$, with compact support, taking value one at zero, and consider
$$
u(t,y,z) = \varphi\left(\frac{\lambda}{2}(\lambda-\sigma)t - \frac{\lambda}{\sigma}y - z\right),
$$
for all $(t,y,z)\in [0,\infty)\times\RR^{2}$. It is easy to check that the above function $u$ satisfies (\ref{eq.reversedHeatEq.mod.BS}).
Let us show that it cannot be represented via (\ref{eq.repres.abstract.deg}). Assume the opposite. Since $\frac{\lambda}{2}(\lambda-\sigma)\neq 0$, there exist $(y,z)\in\RR^2$ and $t>0$, such that $u(t,y,z)=0$ and $u(0,y,z)>0$. Consider
$$
0=u(t,y,z)=\int_{\tilde{\mathcal{E}}} v(t,y,z) \nu(dv).
$$
Since all elements of $\tilde{\mathcal{E}}$ are nonnegative, we conclude that $v(t,y,z)=0$ for $\nu$-almost every $v\in\tilde{\mathcal{E}}$. Next, from the definition of $\tilde{\mathcal{E}}$, we conclude that $v(0,y,z)=0$ for $\nu$-almost every $v\in\tilde{\mathcal{E}}$, and, therefore, $u(0,y,z)=0$. Thus, we obtain the desired contradiction.


The difficulties associated with equation (\ref{eq.reversedHeatEq.mod}) stem from the fact that the operator $\mathcal{L}_{yz}$ is degenerate. The above example shows that this operator may not even be \emph{hypoelliptic}. As a result, the a priori estimates of the solutions to (\ref{eq.reversedHeatEq.mod}), and their derivatives (such as the Schauder estimates and Harnack's inequality), are not readily available. These estimates are crucial for the proofs of Theorems $\ref{th:KoranyiTaylor}$, $\ref{th.strictEllip.rep.abstract}$, and $\ref{th:main}$.
One can, of course, try to restrict the setting by imposing additional conditions on the coefficients of the model, which, although not natural from a financial point of view, may ensure that the operator $\mathcal{L}_{yz}$ satisfies the \emph{H\"ormander condition}, in the sense that the Lie algebra generated by the vector fields from \emph{both the first and the second order differentials} has full rank. The H\"ormander condition yields hypoellipticity of $\mathcal{L}_{yz}$.
See \cite{Kolmogorov}, \cite{Weber}, \cite{Illin}, and \cite{Hormander} for the definitions, existence results, and the construction of fundamental solutions for the equations of H\"ormander type. However, the following example shows that the H\"ormander condition, and, consequently, the hypoellipticity of $\mathcal{L}_{yz}$, is not sufficient for the representation (\ref{eq.repres.abstract.deg}) to be complete.

Consider the following version of (\ref{eq.reversedHeatEq.mod}):
$$
u_t + u_{yy} + y u_z = 0
$$
This is a standard example of a parabolic equation satisfying the H\"ormander condition. In fact, its hypoellipticity was shown in \cite{Kolmogorov}. Notice that the function 
$$
u(t,y,z) = \exp\left(3z - 3ty - 3 t^2\right)
$$
satisfies the above equation. Assume that it can be represented via (\ref{eq.repres.abstract.deg}). Then, using the disintegration, $\mu(d\lambda,d\theta) = \nu(d\lambda,\theta)\rho(d\theta)$, we obtain
$$
e^{3z}=u(0,0,z) = \int_{\RR} e^{\theta z} \nu(\RR,\theta) \rho(d\theta)
$$
From the above, we conclude that $\rho(d\theta)=\delta_3(d\theta)$ and that $\nu(d\lambda,\theta)=\nu(d\lambda)$ is a probability measure on $\RR$. Therefore,
$$
e^{-3t^3} = \int_{\RR} e^{\lambda t} \nu(d\lambda)
$$
is a moment generating function of a probability distribution.
However, Theorem 7.3.5 of \cite{Lukacs} implies that this is impossible.

In fact, it is not surprising that the H\"ormander condition does not resolve our problem: this condition is not sufficient to establish the required a priori estimates, such as the Harnack's inequality, for solutions to (\ref{eq.reversedHeatEq.mod}). For example, the existing forms of Harnack's inequality, available in the literature, require a stronger version of H\"ormander condition, which never holds for the equations of the form (\ref{eq.reversedHeatEq.mod}) (cf. \cite{Kupcov}, \cite{Garofalo} and \cite{Kogoj}).


We have seen that (\ref{eq.repres.abstract.deg}) fails to describe all nonnegative solutions to (\ref{eq.reversedHeatEq.mod}), under the standard assumptions on the model coefficients. Therefore, one can only expect the 'if' part of Theorem (\ref{th.strictEllip.rep.abstract}) to hold true. Such statement would allow us to describe a large (albeit incomplete) class of nonnegative solutions to (\ref{eq.reversedHeatEq.mod}). However, in order to use this result, one would need to know how to construct the elements of $\tilde{\mathcal{E}}$.
The latter may result in a complicated problem on its own, as the associated equation
\begin{equation}\label{eq.temp.1}
(\mathcal{L}_{yz}-\lambda)\psi(y,z)=0
\end{equation}
is degenerate, and it is not immediately clear whether it has a solution and how to compute it. In some particular cases, a change of variables in the above PDE may eliminate the second order derivatives involving $z$ and make the equation similar to (\ref{eq.reversedHeatEq}), with $z$ playing the role of $t$. However, very often, such reduction is not possible, and, even when it is possible, the coefficient in front of $u_z$ may be degenerate, so that we cannot apply Theorems $\ref{th.strictEllip.rep.abstract}$ and $\ref{th:main}$ to characterize the nonnegative solutions of (\ref{eq.temp.1}).
In view of the above discussion, here, we only describe a class of nonnegative solutions to (\ref{eq.reversedHeatEq.mod}), which can be constrcuted by solving a family of uniformly elliptic PDEs (the same level of complexity as the one required to apply Theorem $\ref{th:main}$).


\begin{theorem}\label{th:main.2}
Consider a function $u$, given by
\begin{equation}\label{eq.represent.ext}
u(t,y,z) = \int_{\RR^2} e^{-t\lambda - z\theta} \psi(\lambda,\theta;y) \mu(d\lambda,d\theta),
\end{equation}
for all $(t,y,z)\in\left((0,\infty)\times\RR^{n+1}\right)\cup\left\{(0,0,0)\right\}$,
with a Borel probability measure $\mu$ on $\RR^2$ and a nonnegative function $\psi: \RR^2\rightarrow C^2(\RR^n)$, such that
$\psi\in L^{1}\left(\RR^2\rightarrow C^2(\mathcal{K});\mu\right)$,
for any compact $\mathcal{K} \subset \RR^{n}$ and, for $\mu$-almost every $(\lambda,\theta)$, the following: $\psi(\lambda,\theta;0)=1$ and $\psi(\lambda,\theta;\cdot)$ solves
\begin{equation}\label{eq.elliptic.ext}
\left(\mathcal{L}_y - \theta \sum_{i=1}^n q^i(y) \partial_{y^i} + \theta^2 p(y) - \theta r(y) - \lambda\right) \psi(\lambda,\theta;y)=0,
\end{equation}
for all $y\in\RR^n$.
Then, the function $u$ is a nonnegative classical solution to (\ref{eq.reversedHeatEq.mod}) satisfying $u(0,0,0)=1$.
\end{theorem}

\begin{proof}
The proof is a trivial application of the Hille's (cf. Appendix A or \cite{SwartzFuncAn}) and Fubini's theorems.
\end{proof}


\section{Examples}\label{se:examples}

\subsection{Mean-reverting log-price}\label{subse:ex.meanRev}


Consider a model for the financial market, which consists of only one risky asset $S$ (i.e. $n=k=1$), driven by a one-dimensional Brownian motion $W$ (i.e. $d=1$), via
\begin{equation*}
dS_t = \left(a+\frac{1}{2}\sigma^2 - b \log S_t\right) S_t dt + \sigma S_t dW_t,
\end{equation*}
where $a>0$ and $b>0$ are constants, and, as usual, we assume that the interest rate is zero.
It is easy to see that, in fact, $S$ is the exponential of an Ornstein-Uhlenbeck process.
In particular, we obtain that $Y_t = \log S_t$ satisfies
\begin{equation*}
dY_t = \left(a - b Y_t\right) dt + \sigma dW_t
\end{equation*}
The above model was proposed in \cite{Schwartz} to model the prices of commodities.
Notice that this market model is complete, and, hence, we are in the setting of Subsection $\ref{se:linHJB.complete}$.
Let us describe a family of functions $V:\RR_+\times\RR\times(0,\infty)\rightarrow \RR$, such that $V(t,Y_t,x)$ is a forward performance process. Introducing $u(t,y,z)$, to denote $\left(V_x(t,y,.)\right)^{-1}(\exp(z))$,
we recall that function $u$ is expected to satisfy equation (\ref{eq.illPosedHJB.simplest}), which, in the present setting, becomes
\begin{eqnarray}
&&u_{t} + \frac{1}{2}\left[\frac{1}{\sigma^2}\left(a + \frac{1}{2}\sigma^2 - by\right)^2 u_{zz} -  2 \left(a+\frac{1}{2}\sigma^2-by\right) u_{yz}
+ \sigma^2 u_{yy} \right]\nonumber \\
&&\label{eq.ex1.temp1}\phantom{???????????????????????}+ \frac{ \left(a + \frac{1}{2}\sigma^2 - by\right)^2}{2\sigma^2} u_{z} - \frac{\sigma^2}{2} u_y = 0
\end{eqnarray}
Applying Theorem $\ref{th:main.2}$, we reduce the problem to solving equation (\ref{eq.elliptic.ext}), which, in the present case, becomes
\begin{equation*}
\sigma^2 \psi_{yy} + \left(2\theta \left(a+\frac{1}{2}\sigma^2-by\right)-\sigma^2\right) \psi_y + \left(\theta(\theta-1) \frac{ \left(a + \frac{1}{2}\sigma^2 - by\right)^2}{\sigma^2} - 2\lambda\right)\psi=0
\end{equation*}
It is easy to check that the following functions solve the above ODE, for each $\theta\geq0$,
\begin{equation*}
\psi(\lambda^{\pm},\theta;y) = \exp\left(C^{\pm}_1(\theta)y + C^{\pm}_2(\theta) y^2\right),
\end{equation*}
with the corresponding
\begin{eqnarray*}
&&\lambda=\lambda^{\pm}(\theta) = \theta(\theta-1) \frac{\left(a+\frac{1}{2}\sigma^2\right)^2}{2\sigma^2} + b\left(\theta \pm \frac{1}{2}\sqrt{\theta(3\theta+1)}\right) \\
&&\phantom{????????????????????????}- \frac{ 2a\theta\left(a+\frac{1}{2}\sigma^2\right) + a\sigma^2}{\sigma^2\left(1 \pm \sqrt{3+1/\theta}\right)}
+ \frac{2a^2}{\sigma^2\left(1 \pm \sqrt{3+1/\theta}\right)^2}
\end{eqnarray*}
and
\begin{eqnarray*}
&& C^{\pm}_1 = 1-\frac{2\theta}{\sigma^2}\left(a+\frac{1}{2}\sigma^2\right) - \frac{2a}{\sigma^2\left(1\pm\sqrt{3+1/\theta}\right)},\\
&& C^{\pm}_2 = \frac{b}{2\sigma^2}\left(2\theta \pm \sqrt{\theta(3\theta+1)}\right)
\end{eqnarray*}
According to Theorem $\ref{th:main.2}$, we can construct $u$ via
\begin{eqnarray}
&&\label{eq.ex1.u} u(t,y,z) = \int_{\RR} \exp\left( -z\theta \right) \left[ \exp(C^+_1(\theta) y + C^+_2(\theta) y^2 - t\lambda^+(\theta)) \nu^+(d\theta)
 \right.\\
&&\phantom{???????????????????????}\left.+ \exp(C^-_1(\theta) y + C^-_2(\theta) y^2 - t\lambda^-(\theta)) \nu^-(d\theta)\right],\nonumber
\end{eqnarray}
for arbitrary Borel measures $\nu^+$ and $\nu^-$ on $\RR$, such that the integral
$$
\int_{\RR} e^{-z\theta} \nu^{\pm}(d\theta)
$$
converges for all $z\in\RR$.
Recall that the function $V$ has to be convex in $x$, which implies that the function $u$ needs to be decreasing in $z$. Therefore, we have to restrict measures $\nu^+$ and $\nu^-$ to have support in $\RR_+$.
Notice that the above family does not contain all nonnegative solutions of equation (\ref{eq.ex1.temp1}): in fact, it does not even include all solutions described by Theorem $\ref{th:main.2}$. Nevertheless, it represents a large family of solutions to (\ref{eq.ex1.temp1}) that can be written in a closed form.

Next, we define functions $\tV, V:(0,\infty)\times\RR\times(0,\infty)\rightarrow \RR$ via:
\begin{equation}\label{eq.ex1.tV.def}
\tV(t,y,x) = \left(u(t,y,\log(.))\right)^{-1}(x)\,\,\,\,\,\,\,\,\,\text{and}\,\,\,\,\,\,\,\,\, V(t,y,x) =\int_0^x \tV(t,y,s) ds
\end{equation}
Using the equation (\ref{eq.ex1.temp1}), it is easy to derive a nonlinear PDE for $\tV$ and notice that the same equation arises from a formal differentiation of the HJB equation (\ref{eq.illPosedHJB.complete}) with respect to $x$. However, as it was mentioned in Subsection $\ref{se:linHJB.complete}$, integrating the PDE for $\tV$, to recover the HJB equation (\ref{eq.illPosedHJB.complete}) for $V$, is not always a trivial task and it may require additional arguments. The following proposition takes care of these technical details. Its proof is based on establishing the appropriate estimates for $u$ and $\tV$, and it is given in Appendix B.

\begin{proposition}\label{prop:ex1.1}
For any $a,b,\sigma>0$ and any Borel measures $\nu^+$, $\nu^-$, with compact supports in $(0,\infty)$, the function $V$, given by (\ref{eq.ex1.u})--(\ref{eq.ex1.tV.def}), is well defined and satisfies the HJB equation (\ref{eq.illPosedHJB.complete}), with $n=k=1$, $\mu(y)=a - by$, and $\sigma(y)=\sigma$.
\end{proposition} 

Let us show that $V(t,Y_t,x)$ is a forward performance process. Since $V$ satisfies the HJB equation, it is easy to deduce that, for any portfolio $\pi$, there exists a localizing sequence $\left\{\tau_n\right\}$, such that the process
\begin{equation*}
\left(V(t,Y_{t},X^{\pi,x}_{t})\right)_{t\geq0},
\end{equation*}
stopped at $\tau_n$, is a supermartingale. Function $V$, by construction, is strictly positive, hence, a standard application of Fatou's lemma shows that the above process is a supermartingale itself. Let us now construct the optimal wealth process.
According to (\ref{wealth-sde}), it should satisfy
\begin{eqnarray*}
&&dX_{t}^* = - \frac{1}{\sigma}\left(a + \frac{1}{2}\sigma^2 - bY_t\right) \frac{\frac{1}{\sigma} \left(a + \frac{1}{2}\sigma^2 - bY_t\right) V_x(t,Y_t,X^*_t) + \sigma V_{xy}(t,Y_t,X^*_t) }{V_{xx}(t,Y_t,X^*_t) }  dt \\
&&\phantom{?????????????????}-\frac{\frac{1}{\sigma} \left(a + \frac{1}{2}\sigma^2 - bY_t\right) V_x(t,Y_t,X^*_t) + \sigma V_{xy}(t,Y_t,X^*_t) }{V_{xx}(t,Y_t,X^*_t) } dW_{t}
\end{eqnarray*}
Due to the smoothness of $\tV$, the solution $X^*$ to the above equation is uniquely defined for any initial condition $X^*_0>0$, up to the explosion time.
The estimates (\ref{eq.meanRev.ratioEst}), in turn, imply that the logarithm of $X^*$ (defined, again, up to the explosion time) satisfies:
\begin{eqnarray*}
&&d\log X_{t}^* = \xi_t  dt  + \zeta_t dW_{t},\,\,\,\,\,\,\,\,\,\,
\left|\xi_t\right| \leq c_5(1+Y^2_t),\,\,\,\,\,\,\,\,\,\left|\zeta_t\right| \leq c_5(1+|Y_t|),
\end{eqnarray*}
with a constant $c_3>0$, depending only upon $a$, $b$, $\sigma$ and $\eta$. Since $Y_t$ has finite moments of any order, $X_t$ is square integrable, for any $t$. Hence, $\log(X)$ is a non-exploding continuous process, and, therefore, $X^*$ is strictly positive and non-exploding. The following proposition implies that $V(t,Y_t,x)$ is a forward performance process and, thus, completes the construction. Its proof is given in Appendix B.

\begin{proposition}\label{prop:ex1.2}
The process $\left( V(t,Y_t,X^*_t) \right)_{t\geq0}$ is a martingale.
\end{proposition}


\subsection{Mean-reverting log-volatility}

Here, we consider an example of homothetic forward performance process in a two-factor stochastic volatility model, discussed in Subsection $\ref{se:linHJB.homothetic}$, for which the verification procedure (in particular, the verification of the martingale property) becomes very simple.
Consider a two-factor stochastic volatility model for a single risky asset (i.e. $n=2$ and $k=1$), driven by a two-dimensional Brownian motion $W=(W^1,W^2)$ (i.e. $d=2$), via:
\begin{equation*}
\left\{
\begin{array}{l}
{dS_{t}= S_{t}\left(\kappa - \mu Y_t\right) \exp\left(Y_t\right) dt + S_{t} \exp\left(Y_t\right) dW_{t}^{1},}\\
{dY_t = \left(a - b Y_t\right) dt + \sigma \left(\rho dW^1_t + \sqrt{1-\rho^2} dW^2_t\right),}
\end{array}
\right.
\end{equation*}
where $a\in\RR$, $b>0$, $\kappa\in\RR$, $\mu\geq0$, and $\sigma>0$ are constants. As usual, the interest rate is assumed to be zero.
An additional assumption on $b/\sigma$ is made further in this section.
Notice that the stochastic factor $Y$, in the above model, controls both the spot volatility, $\exp(Y_t)$, and the instantaneous drift. In particular, when the volatility is very large, the drift becomes negative, and vice versa. The stochastic factor itself exhibits a mean-reverting behavior.
As before, we would like to describe a family of functions $V:\RR_+\times\RR\times(0,\infty)\rightarrow \RR$, such that $V(t,Y_t,x)$ is a forward performance process. We make the additional assumption of homothetic preferences:
\begin{equation*}
V(t,y,x) = \frac{x^{\gamma}}{\gamma} v(t,y),
\end{equation*}
for some non-zero constant $\gamma<1$ and function $v:\RR_+\times\RR\rightarrow \RR$ which is yet to be determined.
Thus, we are in the setup of Subsection $\ref{se:linHJB.homothetic}$.
Introducing
\begin{equation*}
u(t,y) =\left(v(t,y)\right)^{1/\delta},\,\,\,\,\,\,\text{with}\,\,\,\,
\delta=\frac{1-\gamma}{1-\gamma +\rho^2\gamma},
\end{equation*}
we notice that, in this case, equation (\ref{eq.illPosedHJB.homothetic}) becomes
\begin{equation*}
u_{t} + \frac{1}{2}\sigma^2 u_{yy} + \left( a - by +\rho\sigma \frac{\gamma }{1-\gamma }(\kappa-\mu y) \right) u_{y}
+ \frac{1}{2\delta }\frac{\gamma }{1-\gamma }(\kappa-\mu y)^2 u=0
\end{equation*}
Applying Theorem $\ref{th:main}$, we reduce the problem to equation (\ref{eq.elliptic}), which, in the present case, becomes
\begin{equation*}
\frac{1}{2}\sigma^2 \psi_{yy} + \left(a - by +\rho\sigma \frac{\gamma }{1-\gamma } (\kappa-\mu y)\right) \psi_y
+ \left(\frac{1}{2\delta }\frac{\gamma }{1-\gamma }(\kappa-\mu y)^2 - \lambda\right)\psi=0
\end{equation*}
It is, then, easy to check that the following functions
\begin{equation*}
\psi(\lambda^{\pm};y) = \exp\left(C^{\pm}_1 y + C^{\pm}_2 y^2\right),
\end{equation*}
solve the above ODE, with the corresponding
\begin{equation*}
\lambda^{\pm}=\sigma^2\left(\frac{1}{2}\left(C^{\pm}_1\right)^2 + C_2^{\pm}\right) + C_1^{\pm}\left(a+\rho\sigma\kappa\frac{\gamma}{1-\gamma}\right)
+ \frac{1}{2\delta}\frac{\gamma}{1-\gamma} \kappa^2
\end{equation*}
and
\begin{eqnarray*}
&& C_1^{\pm} = \pm\frac{\frac{\kappa\mu}{\sigma}\frac{\gamma}{1-\gamma}\left(1+\frac{\rho^2\gamma}{1-\gamma}\right)
- 2C_2^{\pm} \left(\frac{a}{\sigma} + \kappa\rho\frac{\gamma}{1-\gamma}\right) }{\sqrt{\left(\frac{b}{\sigma} + \mu \rho\frac{\gamma}{1-\gamma}\right)^2 - \frac{\mu^2}{\delta}\frac{\gamma}{1-\gamma}}},\\
&& C_2^{\pm} = \frac{1}{2}\left(\frac{b}{\sigma} + \mu \rho\frac{\gamma}{1-\gamma}\right)\pm
\frac{1}{2}\sqrt{\left(\frac{b}{\sigma} + \mu \rho\frac{\gamma}{1-\gamma}\right)^2 - \frac{\mu^2}{\delta}\frac{\gamma}{1-\gamma}},
\end{eqnarray*}
where it is assumed that
\begin{equation}\label{eq.ex1.assump}
\frac{b}{\sigma} \geq \mu\left(\sqrt{\rho^2\frac{\gamma^2}{(1-\gamma)^2} + \frac{\gamma}{1-\gamma}} - \rho \frac{\gamma}{1-\gamma} \right)
\end{equation}
In particular, the function
\begin{equation*}
u(t,y) = \nu^+ e^{-t\lambda^+} \exp\left(C^{+}_1 y + C^{+}_2 y^2\right) +
\nu^- e^{-t\lambda^-} \exp\left(C^{-}_1 y + C^{-}_2 y^2\right)
\end{equation*}
solves (\ref{eq.illPosedHJB.homothetic}), and, therefore, the following function is a solution to the forward HJB equation (\ref{eq.illPosedHJB}):
\begin{equation*}
V(t,y,x) = \frac{x^{\gamma}}{\gamma}\left( \nu^+ e^{-t\lambda^+} \exp\left(C^{+}_1 y + C^{+}_2 y^2\right) + \nu^- e^{-t\lambda^-} \exp\left(C^{-}_1 y + C^{-}_2 y^2\right) \right)^{\delta},
\end{equation*}
for arbitrary $\nu^+,\nu^-\geq0$.
As in the previous example, it is straightforward to check that, for any portfolio $\pi$, the process $\left(V(t,Y_{t},X^{\pi,x}_{t})\right)_{t\geq0}$ is a supermartingale. The equation for the optimal wealth process becomes
\begin{equation}
dX_{t}^* = \frac{X^*_t}{1-\gamma} \left(\kappa-\mu Y_t\right) \left( \kappa-\mu Y_t  + \sigma\rho
\frac{ u_y\left(t,Y_t\right) }{ u\left(t,Y_t\right) } \right) dt 
+ \frac{X^*_t}{1-\gamma} \left( \kappa-\mu Y_t  + \sigma\rho
\frac{ u_y\left(t,Y_t\right) }{ u\left(t,Y_t\right) }  \right) dW^1_{t}\label{eq.ex2.Xstar}
\end{equation}
It is easy to see that
\begin{equation}\label{eq.ex2.uy}
\left|\frac{ u_y\left(t,y\right) }{ u\left(t,y\right) } \right| \leq c_6 (1+|y|)
\end{equation}
Hence, we conclude that, for any initial condition $X^*_0>0$, the equation (\ref{eq.ex2.Xstar}) has a unique strong solution $X^*$ which is strictly positive. To show that $V(t,Y_t, x)$ is a forward performance process, it only remains to apply the following proposition, whose proof is given in Appendix B.

\begin{proposition}\label{prop:ex2.1}
The process $\left( V(t,Y_t,X^*_t) \right)_{t\geq0}$ is a martingale.
\end{proposition}

\begin{remark}
It is worth mentioning that the optimal wealth process, defined by (\ref{eq.ex2.Xstar}), is monotone in the initial wealth. This observation shows that the forward performance process constructed in this example belongs to the class of processes characterized in \cite{ElKaroui.2}. In fact, it is easy to see that the same is true for any homothetic forward performance process, defined in Subsection $\ref{se:linHJB.homothetic}$. 
As discussed in the introduction, this paper does not aim to generalize the space of forward performance processes, and, in particular, we do not consider more general processes than those studied in \cite{ElKaroui.2}. Instead, this work provides a new, convenient, representation of a large class of these random fields. Namely, the representation provided herein allows one to start with the economically meaningful input elements (the stochastic factor $Y$ and the investor's initial preferences $U_0$) and determine the associated forward performance process, from this input, uniquely.
\end{remark}


\section{Summary}

We have described a new approach to constructing investment strategies with optimal payoffs at all positive time horizons, where the associated optimality criteria are given by the forward investment performance processes. We outlined the main difficulties associated with the construction of the forward performance processes and summarized the existing results in this direction.

We, then, demonstrated that the theory of forward performance admits an axiomatic justification, in the spirit of classical expected utility theory. 
Motivated by the axiomatic approach, we proposed a new representation of the forward performance processes, using the parameters that have direct economic interpretation. 
In a Markovian setting, the proposed representation lead us to the analysis of forward investment performance processes in a factor form.

We characterized the forward performance processes in a factor form via solutions to a time-reversed HJB equation. In the case when this equation can be linearized, we obtained an explicit integral representation of its nonnegative solutions. In particular, our results allow to construct the forward performance process in a factor form (explicitly, or as a numerical solution to a standard elliptic PDE), given its initial value (the investor's initial preferences) and a diffusion model for the associated stochastic factor.

In the course of our study, we have obtained a generalization of Widder's theorem on the representation of all positive solutions to a time-reversed parabolic PDE on a semi-infinite time interval. In order to do this, we combined the existing characterization of the minimal elements of the space of all positive solutions with some basic facts from Potential Theory and Convex Analysis. From a probabilistic point of view, our results provide a representation of the Martin boundary of a space-time diffusion via the Martin boundary of the diffusion process itself.

Further research should address the problem of solving the time-reversed HJB equation itself. In addition to all the difficulties associated with the standard HJB equation, this problem is ill-posed, as it has ``time running in a wrong direction". This feature makes it very hard to determine the initial conditions for which the solutions exist, as well as to find a tractable description of the resulting solutions.

Another related problem is the calibration of a forward performance process to the investor's initial preferences. 
Our study shows that, in many cases, the forward performance process is uniquely determined by its value at time zero. We have seen that the latter should be interpreted as a state dependent utility function which describes the investor's preferences at a short time horizon. In order to complete the analysis, it is important to develop a reliable algorithm for determining this function from investor's choices.

\section{Appendix A}

In this appendix, we recall some standard technical results.

\subsection{Parabolic PDE}
Firstly, we are interested in quantitative properties of the solutions to the parabolic PDE (\ref{eq.reversedHeatEq}), with the differential operator $\mathcal L_y$ defined in (\ref{eq.Ly.def}) and in the subsequent paragraph. 
We make use of the following version of Harnack's inequality.
\begin{theorem}(Harnack's inequality)
Suppose $u$ is a nonnegative solution to (\ref{eq.reversedHeatEq}) in $(0,\infty)\times\RR^n$.
Then, for any $R>0$, there exists a constant $C(R) > 0$, depending only on $R$, on the upper bounds of the absolute values of the coefficients in $\mathcal{L}_y$, and on the lower and upper bounds of the associated quadratic form, such that
$$
\sup_{ \|y\| \leq 1} u(R, y) \leq C(R) u(0,0).
$$
\end{theorem}
\begin{proof} This statement follows immediately from Theorem 1.1 of \cite{HarnackUniform}, after time reversal and shifting the space variable, in the PDE considered in \cite{HarnackUniform}.
\end{proof} 

The second result which is needed repeatedly is a version of the interior Schauder estimate.
Define the H\"older norm on a domain $D \subseteq \RR^{1+n}$ by
$$
\| v \|_{D, \alpha} = \sup_{(t,y) \in D} | v(t,y) | + \sup_{(s,x), (t,y) \in D} \frac{\| v(t,y) - v(s,x)\|}{\|y-x\|^{\alpha} + |t-s|^{\alpha/2}}.
$$ 
For $\varepsilon > 0$ and $T > 0$, let  
$$
D_{\varepsilon}^T = \{ (t,y): \varepsilon \| y\|^2 \le t \le T \}.
$$

\begin{theorem}[Interior Schauder estimate]
Assume that the coefficients of $\mathcal L_y$ are H\"older-continuous with the H\"older exponent $0 < \alpha < 1$. Then for any positive $\varepsilon, T$ and $\delta$, there 
exists a constant $C > 0$, depending on $\varepsilon, T, \delta$, and on the coefficients of $\mathcal L_y$,
such that
$$
\| u \|_{D_{\varepsilon}^T, \alpha} + \delta^{(1+\alpha)/2} \| \partial_y u \|_{D_{\varepsilon}^T, \alpha}  + \delta^{1+\alpha/2} \| \partial^2_y u \|_{D_{\varepsilon}^T, \alpha}  
+ \delta^{1+\alpha/2} \| u_t \|_{D_{\varepsilon}^T, \alpha} 
\le C  \sup_{(t,y) \in D_{\varepsilon}^{T+\delta}} | u(t,y) |.
$$
\end{theorem}

\begin{proof} See the article of Knerr \cite{Schauder}.
\end{proof}

\subsection{Elliptic PDE}
We now consider the question of positive solutions of the elliptic equation (\ref{eq.elliptic}), with the differential operator $\mathcal L_y$ defined in (\ref{eq.Ly.def}).
\begin{theorem}
If the operator $\mathcal L_y - \lambda $ has a Green's function then equation (\ref{eq.elliptic}) has a positive solution.
A sufficient condition for the existence of a Green's function is 
$$
\int_0^{\infty} \EE_x\left[ e^{ \int_0^t  c(X_s)ds - \lambda t} \right] dt  < \infty,
$$
for all $x \in \RR^n$, where $(X_t)_{t \ge 0}$ is the diffusion with generator
$$
\mathcal L^0_y = \mathcal L_y - c(y).
$$
\end{theorem}
\begin{proof}   See Theorems $3.1$ and $3.6$ in Section $4.3$ of Pinsky's book \cite{Pinsky}. 
\end{proof}

\subsection{Vector integration}
Now we recall the construction of the Bochner integral as needed in Section $\ref{se:timeReversedHeat}$. Let $(F, \mathcal F, \mu)$ be a measurable space (with a finite measure $\mu$) and let $B$ be a Banach space with norm $\| \cdot \|$.  For simple functions of the form
$$
g = \sum_{i=1}^N b_i \bone_{F_i}
$$
where $F_i \in \mathcal F$ and $b_i \in B$ for each $i$, we let
$$
\int g \ d\mu =  \sum_{i=1}^N  b_i \mu(F_i).
$$
To define the Bochner integral of a general function $g: F\rightarrow B$, we consider a sequence of simple functions $g_n$ such that
$$
\int \| g- g_n \|  d\mu \to 0,
$$
as $n \to \infty$.
Then, the integral $\int_F g d\mu$ is defined as the limit of the sequence of integrals $\int_F g_n \ d\mu$, which converges in the strong topology of $B$.
It is easy to show (cf. \cite{SwartzFuncAn}) that, whenever $\int \| g  \|  d\mu  < \infty$, such sequence of simple functions $g_n$ does exist, and the limit of $\int g_n \ d\mu$ depends only on the function $g$, but 
not on the particular choice of the sequence.

Like the Lebesgue integral, the Bochner interal is rather robust.  A particular instance of this 
robustness is that we can interchange integration and linear functionals.
\begin{theorem}(Hille)
Let $g$ be a Bochner integrable function and $T:B\rightarrow\RR$ be a continuous linear functional. Then
$$
T \int_F g d\mu = \int_F T(g) d\mu
$$
\end{theorem}
\begin{proof} This result can also be found \cite{SwartzFuncAn}.
\end{proof}

\section{Appendix B}

\subsection{Proof of Proposition $\ref{prop:ex1.1}$} 
First, using the definition of $\tV$ and equation (\ref{eq.ex1.temp1}), we obtain the following PDE for $\tV$: 
\begin{eqnarray}
&&\label{eq.complete.Vx}
\tV_t+ \frac{1}{2}\sigma^2 \tV_{yy}  + (a - b y)\tV_y
+ \frac{1}{2} \frac{\tV_{xx}}{\tV^2_{x}} \left(\sigma \tV_{xy} + \frac{a + \sigma^2/2 - by}{\sigma} \tV_x\right)^2\\
&&\phantom{????????????????} - \frac{1}{\tV_x}\left(\sigma \tV_{xy} + \frac{a + \sigma^2/2 - by}{\sigma}\tV_x \right) 
\left(\sigma \tV_y + \frac{a + \sigma^2/2 - by}{\sigma} \tV \right) = 0\nonumber
\end{eqnarray}
It is a standard exercise to check that the left hand side of the above is the $x$-derivative of the left hand side of the HJB equation (\ref{eq.illPosedHJB.complete}), with $V$ given by (\ref{eq.ex1.tV.def}). Thus, in order to prove that $V$ solves (\ref{eq.illPosedHJB.complete}), it only remains to show that the value of the left hand side of (\ref{eq.illPosedHJB.complete}), with $V$ given by (\ref{eq.ex1.tV.def}), converges to zero, as $x\downarrow0$.
For this, we need to establish the appropriate estimates of the partial derivatives of $\tV$ and, in turn, of $V$.

Assume that the measures $\nu^+$ and $\nu^-$ have supports in $[1+\eta,1/\eta]$, for some $\eta\in(0,1/2)$, and at least one of these measures is not identically zero (if they are both zero, then, the statement is obvious).
It follows from (\ref{eq.ex1.u}) that there exists $c_1=c_1(t,y)\in(0,1)$, which is a continuous function of $(t,y)\in\RR_+\times\RR$, such that
\begin{equation*}
 c_1(t,y) \left(x^{-1-\eta}\wedge x^{-1/\eta}\right) \leq u(t,y,\log(x)) \leq \frac{1}{c_1(t,y)} \left(x^{-1-\eta}\vee x^{-1/\eta}\right),
 \,\,\,\,\,\,\,\,\forall x>0
\end{equation*}
This yields
\begin{equation}\label{eq.meanRev.tV.est}
\tV(t,y,x) \leq c_1^{-1/(1+\eta)}(t,y) x^{-1/(1+\eta)} + c_1^{-\eta}(t,y) x^{-\eta},
\,\,\,\,\,\,\,\,\,\,\forall (t,y,x)\in\RR_+\times\RR\times(0,\infty)
\end{equation}
It is also easy to see, using (\ref{eq.ex1.u}), that there exists $c_2>0$, depending only upon $a$, $b$, $\sigma$ and $\eta$, such that
\begin{equation*}
\eta\leq-\frac{u(t,y,z)}{u_z(t,y,z)} \leq \frac{1}{1+\eta}\,\,\,\,\,\text{and}\,\,\,\,\,\,\, \left|\frac{u_y(t,y,z)}{u(t,y,z)}\right| \leq c_2\left(1+|y|\right)
\end{equation*}
hold for all $(t,y,z)\in \RR_+\times\RR^2$.
It follows that
\begin{equation}\label{eq.meanRev.ratioEst}
(1+\eta) x\leq-\frac{\tV(t,y,x)}{\tV_{x}(t,y,x)} \leq \frac{1}{\eta} x,\,\,\,\,\,\text{and}\,\,\,\,\,\,\, \left|\frac{\tV_{y}(t,y,x)}{\tV_{x}(t,y,x)}\right| \leq c_2 \left(1+|y|\right) x
\end{equation}
Similarly, we deduce that
\begin{equation*}
\left|\frac{u_{zz}(t,y,z)}{u_z(t,y,z)}\right| \leq \frac{1}{\eta}\,\,\,\,\,\text{and}\,\,\,\,\,\,\, \left|\frac{u_{yy}(t,y,z)}{u(t,y,z)}\right| \leq c_3\left(1+y^2\right),
\end{equation*}
where $c_3>0$ depends only upon $a$, $b$, $\sigma$ and $\eta$.
Next, we recall from (\ref{eq.ex1.tV.def}) that
\begin{equation*}
e^{-z}\tV_{yy}\left(t,y,u(t,y,z)\right) =  -\frac{u_y^2}{u_z^2} \frac{u_{zz}-u_z}{u_z} + 2 \frac{u_y}{u_z} \frac{u_{yz}}{u_z} - \frac{u_{yy}}{u_z},
\end{equation*}
to obtain
\begin{equation}\label{eq.meanRev.tVyy.est}
\left|\tV_{yy}\left(t,y,x\right)\right| \leq c_4 (1+y^2) x,\,\,\,\,\,\,\,\,\,\,\forall (t,y,x)\in\RR_+\times\RR\times(0,\infty),
\end{equation}
where $c_4>0$ depends only upon $a$, $b$, $\sigma$ and $\eta$. The estimates (\ref{eq.meanRev.tV.est}), (\ref{eq.meanRev.ratioEst}) and (\ref{eq.meanRev.tVyy.est}), along with the Fubini's theorem, imply that $V(t,y,x)$ is well defined, with its $y$-derivatives are given by:
\begin{equation*}
V_y(t,y,x) =\int_0^x \tV_y(t,y,s) ds,\,\,\,\,\,\,\,\, V_{yy}(t,y,x) =\int_0^x \tV_{yy}(t,y,s) ds.
\end{equation*}
Applying the same estimates and the Fubini's theorem again, we conclude that the right hand side of (\ref{eq.illPosedHJB.complete}), with $V$ given by (\ref{eq.ex1.tV.def}), converges to zero, as $x\downarrow0$. This completes the proof of the proposition.

\subsection{Proof of Proposition $\ref{prop:ex1.2}$} 
Recall, from the results discussed in Subsection $\ref{subse:forwPerf}$, that the process $\left(V(t,Y_{t},X^{*}_{t})\right)_{t\geq0}$ is a local martingale. Let us show that it is, in fact, a true martingale.
Applying the It\^o's lemma, we obtain
\begin{equation*}
d\log V(t,Y_{t},X^{*}_{t})= -\frac{1}{2}Z^2_t dt + Z_t dW_t,
\end{equation*}
where
\begin{equation*}
Z_t = \sigma \frac{V_y(t,Y_t,X^*_t)}{V(t,Y_t,X^*_t)} - \frac{\tV(t,Y_t,X^*_t)}{V(t,Y_t,X^*_t)} \frac{\frac{1}{\sigma} \left(a + \frac{1}{2}\sigma^2 - bY_t\right) \tV(t,Y_t,X^*_t) + \sigma \tV_{y}(t,Y_t,X^*_t) }{\tV_{x}(t,Y_t,X^*_t) }
\end{equation*}
Applying (\ref{eq.meanRev.ratioEst}), we obtain:
\begin{eqnarray*}
&&V(t,y,x) \leq -\frac{1}{\eta} \int_0^x s \tV_{x}(t,y,s) ds = -\frac{1}{\eta} x \tV(t,y,x) + \frac{1}{\eta} V(t,y,x)
\,\,\,\,\,\Rightarrow\,\,\,\,\frac{\tV(t,y,x)}{V(t,y,x)} \leq \frac{1-\eta}{x} ,\\
&& \left|V_{y}(t,y,x)\right| \leq -c_2 \left(1+|y|\right) \int_0^x s \tV_{x}(t,y,s) ds = - c_2 \left(1+|y|\right) x \tV(t,y,x)
+ c_2 \left(1+|y|\right) V(t,y,x)\\
&&\,\,\,\,\,\, \Rightarrow\,\,\,\,\,\,\,\left|\frac{V_y(t,Y_t,X^*_t)}{V(t,Y_t,X^*_t)}\right| \leq c_2 \left(1+|y|\right)
\end{eqnarray*}
The above inequalities and (\ref{eq.meanRev.ratioEst}) imply that
\begin{equation}\label{eq.ex1.Zest}
|Z_t|\leq c_6 \left(1+|Y_t|\right)
\end{equation}
Next, we use the Novikov's condition (more precisely, the ``salami" method, given, for example, in Corollary 5.14 in \cite{KarShreve}) to conclude that $V(t,Y_{t},X^{*}_{t})$ is a true martingale. According to this method, we only need to verify that, for any $T>0$, there exists $\Delta>0$, such that
\begin{equation*}
\EE \exp\left(\frac{1}{2}\int_{t}^{t+\Delta} Z^2_s ds\right) <\infty,
\end{equation*}
for all $t\in[0,T]$. Using (\ref{eq.ex1.Zest}) and the representation of an Ornstein-Uhlenbeck process as a time-changed Brownian motion, we obtain
\begin{eqnarray*}
&&\exp\left(\frac{1}{2}\int_{t}^{t+\Delta} Z^2_s ds\right) \leq c_7 \exp\left(\frac{1}{2}\int_{t}^{t+\Delta} Y^2_s ds\right)\\
&&\phantom{??????????}\leq c_8 \exp\left(c_9 \int_{t}^{t+\Delta} W^2_{\exp(2bs)-1} e^{-bs} ds\right)
\leq c_8 \exp\left(c_9 \Delta \sup_{s\in[0,\exp(2bT)]} W^2_{s}\right)
\end{eqnarray*}
It is easy to see that we can choose $\Delta>0$ small enough, so that the right hand side of the above is integrable. This completes the construction.

\subsection{Proof of Proposition $\ref{prop:ex2.1}$}
Applying the It\^o's formula, we obtain
\begin{equation*}
d\log V(t,Y_{t},X^{*}_{t})= -\frac{1}{2}\left(Z^2_t + N^2_t\right) dt + Z_t dW^1_t + N_t dW^2_t,
\end{equation*}
where
\begin{eqnarray*}
&& Z_t := \sigma\rho \frac{ u_y\left(t,Y_t\right) }{ u\left(t,Y_t\right) }  + \frac{\gamma}{1-\gamma} \left(\kappa-\mu Y_t  + \sigma\rho
\frac{ u_y\left(t,Y_t\right) }{ u\left(t,Y_t\right) }  \right),\,\,\,\,\,\,\,\, 
N_t = \sigma\sqrt{1-\rho^2}\delta \frac{ u_y\left(t,Y_t\right) }{ u\left(t,Y_t\right) }
\end{eqnarray*}
The estimate (\ref{eq.ex2.uy}) yields $|Z_t| + |N_t| \leq c_7 \left(1+|Y_t|\right)$.
Repeating the last argument in the proof of Proposition $\ref{prop:ex1.2}$, given above, we conclude that $V(t,Y_{t},X^{*}_{t})$ is, indeed, a true martingale.


\bibliographystyle{plain}
\bibliography{ForwardHJB_lin_refs}

\end{document}